\documentclass[twocolumn,secthm]{autart}

\usepackage{graphicx}
\usepackage{url}
\usepackage{amsmath,amssymb}
\usepackage{bm}
\usepackage{mathtools}
\usepackage{algorithm}
\usepackage{algpseudocode}
\usepackage{enumitem}

\newtheorem{theorem}[thm]{Theorem}
\newtheorem{assumption}[thm]{Assumption}
\newtheorem{remark}[thm]{Remark}
\newtheorem{lemma}[thm]{Lemma}
\newtheorem{corollary}[thm]{Corollary}
\newtheorem{definition}[thm]{Definition}
\newenvironment{proof}{\begin{pf*}{Proof.}}{\hfill$\blacksquare$\end{pf*}}

\newcommand{\cK}{\mathcal{K}}

\newcommand{\Kop}{\mathcal{K}}

\newcommand{\bGamma}{\boldsymbol{\Gamma}}

\begin{document}

\begin{frontmatter}

\title{A Koopman--Polynomial Chaos Framework for SMPC of Nonlinear
Stochastic Processes\thanksref{footnoteinfo}}

\thanks[footnoteinfo]{$^*$Authors contributed equally.
$^\dagger$Corresponding author:
% Liang Wu, wliang14@jh.edu
Richard D. Braatz, braatz@mit.edu
.}
\thanks[conferenceinfo]{A preliminary version of this work will be presented at the
2026 IEEE Conference on Decision and Control (CDC).}
\author[MIT]{Wallace Gian Yion Tan$^*$}
\author[MIT]{Efstathios Iliakis$^*$}
\author[JHU]{Liang Wu}
\author[JHU]{J\'an Drgo\v{n}a}
\author[MIT]{Richard D. Braatz$^\dagger$}\ead{braatz@mit.edu}

\address[MIT]{Massachusetts Institute of Technology, MA 02139, USA.}
\address[JHU]{Johns Hopkins University, MD 21218, USA.}


\begin{abstract}
This paper develops a polynomial parametric Koopman operator (PPKO) framework for stochastic model predictive control (SMPC) of nonlinear systems with parametric uncertainty and time-varying stochastic disturbances. The Koopman operator is parameterized by polynomial chaos expansions (PCEs), enabling data-driven nonlinear dynamics to be integrated with PCE-based uncertainty propagation and control. We show that the polynomial parameterization introduces no irreducible approximation gap beyond PCE truncation. We then derive a condensed deterministic SMPC reformulation that handles both time-invariant parametric uncertainty and time-varying disturbances while retaining an online optimization dimension that depends only on the control horizon and input dimension, independent of the lifted-state and PCE dimensions. The resulting formulation yields convex programs with expectation and second-moment constraints that can be solved using standard optimization methods. Numerical studies on the Duffing oscillator and the electrothermal cooling system demonstrate the accuracy, scalability, and closed-loop performance of the proposed framework.
\end{abstract}

\end{frontmatter}

\section{Introduction}
Many real-world control systems operate under significant uncertainties arising from uncertain physical parameters and varying operating conditions. Some common examples include chemical reaction systems with uncertain kinetic and heat transfer coefficients \cite{ruppen1995optimization}, mechanical systems with uncertain mass or stiffness parameters \cite{mammarella2018offline}, and fluctuations in ambient temperatures in building control \cite{hewing2020recursively}. In such scenarios, stochastic model predictive control (SMPC) has emerged as a principled approach to handle uncertainty by explicitly incorporating probabilistic models of uncertain parameters and stochastic disturbances into the predictive control framework \cite{mesbah2016stochastic}. 

A fundamental challenge in the practical implementation of SMPC is the need to propagate uncertainty through the system dynamics and evaluate expectations appearing in the cost and constraint functions. When the underlying system dynamics are nonlinear, the resulting SMPC problem is typically formulated as a nonlinear program (NLP), which must be solved at each sampling time. Such NLPs are generally computationally expensive and difficult to integrate with real-time optimization in practical applications. In contrast, for stochastic systems with linear dynamics, a rich set of computationally efficient methods exists in the literature, including approaches based on polynomial chaos expansions (PCEs) \cite{tan2025offset,dai2016distributed,bernardini2011stabilizing,mishra2024polynomial,ou2025polynomial,wan2021polynomial}. These PCE-based methods approximate the stochastic dynamics by projecting onto a set of orthogonal polynomials, which results in a deterministic expanded system of PCE coefficients that can be controlled with deterministic MPC methods. 

Although PCE-based SMPC methods are widely used, they have several notable drawbacks. First, stochastic Galerkin projections (SGPs) do not necessarily preserve stability, even for linear systems, and introduce truncation errors that can propagate over the prediction horizon \cite{pulch2019stability,evangelisti2024application}. Second, PCE-based methods suffer from the curse of dimensionality, as the number of basis polynomials for a fixed degree increases combinatorially with the uncertainty dimension \cite{blatman2011adaptive}. This issue is further exacerbated in the time-varying disturbance setting, where the uncertainty dimension increases linearly with the SMPC horizon, further worsening the curse of dimensionality \cite{bhusal2021generalized,bradford2019output}. On the other hand, in the deterministic MPC literature, the MPC problem dimension can be made to depend only on the horizon and control dimensions through a technique known as model \textit{condensation} \cite{rawlings2020model,jerez2012sparse}. This process algebraically eliminates the states by recursively substituting the dynamical equations into the prediction model. Altogether, these observations motivate the use of condensing techniques in PCE-based SMPC to ensure that the online optimization problem dimension scales only with the horizon and control dimensions, rather than with the number of PCE basis polynomials.

Recently, a popular data-driven approach to improve the computational tractability of nonlinear deterministic MPC is to learn a linear high-dimensional representation of the dynamics with Koopman operator methods~\cite{korda2018linear,proctor2018generalizing,arbabi2018data,shi2022deep,lazar2024basis,williams2016extending}. These methods enable the reformulation of nonlinear MPC optimization problems as compact QPs after condensation to remove the high-dimensional state dependence \cite{wu2026least}. To address modeling errors arising from the finite-dimensional approximation of the Koopman operator, robust tube-based MPC formulations have been proposed that ensure closed-loop stability under bounded additive disturbances~\cite{zhang2022robust}, while stochastic MPC has been combined with Koopman models to handle approximation uncertainty probabilistically through chance constraints~\cite{kim2025ksmpc}. However, existing approaches do not jointly incorporate time-invariant parametric uncertainty and time-varying stochastic disturbances directly into the Koopman representation in a form that enables PCE-based uncertainty propagation and condensed SMPC. Parametric Koopman operators have been proposed to extend Koopman representations to systems with time-invariant or time-varying inputs~\cite{guo2025learning,williams2016extending}, which implicitly suggests parameterizing the Koopman operator by a PCE. This approach enables the application of linear SMPC methods to nonlinear systems, which is the gap addressed by our proposed Polynomial Parametric Koopman Operators (PPKO) framework.
\subsection{Contributions}
The main contributions of this article are as follows.
\begin{itemize}
\item We propose a PCE-based parametric Koopman operator framework for nonlinear dynamical systems with time-invariant parametric uncertainty and time-varying stochastic disturbances. 
\item We show that restricting the parameter dependence of the Koopman operator to a finite-order PCE introduces no irreducible approximation gap beyond polynomial truncation; under the stated assumptions, this truncation error vanishes as the polynomial order increases.
\item We derive a condensed deterministic reformulation of the SMPC problem for systems with both time-invariant parametric uncertainty and time-varying stochastic disturbances. Exploiting the IID disturbance structure, the required first- and second-order moments are propagated recursively and computed offline, yielding an online convex optimization problem whose dimension scales only with the control horizon and input dimension, independent of the number of Koopman observables and retained PCE terms.
\item We validate the framework on both ODE and PDE examples, including the Duffing oscillator and the electrothermal cooling system.
\end{itemize}

This article extends the previously presented results in \cite{iliakis2026polynomial} by generalizing the framework to include time-varying stochastic disturbances, a rigorous analysis of the approximation gap when parameterizing Koopman operators with PCEs, and an additional numerical study using the electrothermal cooling system as a high-dimensional PDE--ODE example.

\section{Problem Formulation}
Consider the problem of regulating nonlinear stochastic discrete-time systems to the origin:
\begin{equation}\label{eq:nonlinear}
    x_{t+1} = f(x_t,u_t,w_t,\theta),
\end{equation}
where $x_t \in \mathbb{X} \subseteq \mathbb{R}^{n_x}$ is the state, $u_t \in \mathbb{U} \subseteq \mathbb{R}^{n_u}$ is the control input, $w_t \in \mathbb{W} \subseteq \mathbb{R}^{n_w}$ is a vector of time-varying disturbances, and $\theta \in \Theta \subseteq \mathbb{R}^{n_\theta}$ is a random vector characterizing parametric uncertainty. We impose the following assumptions:
\begin{assum}\label{assum:system}
The sets $\mathbb{X} \subset \mathbb{R}^{n_x}$, $\mathbb{U} \subset \mathbb{R}^{n_u}$, $\mathbb{W} \subset \mathbb{R}^{n_w}$, and $\Theta \subset \mathbb{R}^{n_\theta}$ are Borel sets, and the map $f:\mathbb{X}\times\mathbb{U}\times\mathbb{W}\times\Theta \to \mathbb{X}$ is jointly continuous. The state space $\mathbb{X}$ is forward invariant under \eqref{eq:nonlinear}, i.e., $f(x,u,w,\theta)\in\mathbb{X}$ for all $x\in\mathbb{X}$, $u\in\mathbb{U}$, $w\in\mathbb{W}$, and $\theta\in\Theta$. The input set $\mathbb{U}$ is compact and contains the origin. The state space $\mathbb{X}$ contains the origin and may be unbounded. The admissible input set is $\mathbb{U}:=\{u\in\mathbb{R}^{n_u}:u_{\min}\le u\le u_{\max}\}$, where the inequalities are interpreted component-wise.
\end{assum}
\begin{assum}
The disturbance sequence $\{w_t\}_{t=0}^\infty$ is assumed to be i.i.d. and independent of $\theta$. We assume that the probability laws of $w_t$ and \(\theta\) are known and that their realizations are not directly measured.
\end{assum}
% \begin{assum}\label{assum:equilibrium}
% The origin is an equilibrium of \eqref{eq:nonlinear} in the absence of disturbance, i.e., for all $\theta \in \Theta$, one has 
% $
% f(0,0,0,\theta)=0.$
% \end{assum}
Stochastic Model Predictive Control (SMPC) addresses this problem by solving a finite-horizon stochastic optimal control problem with additional state and control constraints. Assuming that a measurement of the state $x_0$ is available at time $ t = 0$, the finite-horizon moment-constrained SMPC problem with a linear quadratic regulator (LQR) type cost is shown in \eqref{eq:smpc_problem}.
\noindent\rule{\columnwidth}{0.4pt}
\vspace{-0.7em}
\textbf{Stochastic MPC}
\vspace{-0.4em}

{\small
\setlength{\abovedisplayskip}{2pt}
\setlength{\belowdisplayskip}{2pt}
\setlength{\jot}{2pt}
\begin{subequations}\label{eq:smpc_problem}
\begin{align}
\min_{\{u_t\}_{t=0}^{H-1}} 
&~ \sum_{t=0}^{H-1} 
\mathbb{E}\big[x_t^\top Q x_t + u_t^\top R u_t \big]
+ \mathbb{E}\big[x_H^\top Q_f x_H\big]
\label{eq:smpc-cost}
\\
\text{s.t.}~
& x_{t+1} = f(x_t,u_t,w_t,\theta),\quad t=0,\dots,H-1,
\label{eq:smpc-dyn}
\\
& x_t^{\min} \le \mathbb{E}[x_t] \le x_t^{\max},
\quad t=1,\dots,H,
\label{eq:smpc-mean}
\\
& \mathbb{E}\big[(a_t^\top x_t-b_t)^2\big] \le c_t^2,
\quad t=1,\dots,H,
\label{eq:2ndorderconstraints}
\\
& u_{\min} \le u_t \le u_{\max},
\quad t=0,\dots,H-1
\label{eq:smpc-input}
\end{align}
\end{subequations}
}
\par\noindent\rule{\columnwidth}{0.4pt}\par
Here, $H$ denotes the prediction horizon; the matrices $Q \succeq 0$, $R \succ 0$, and $Q_f \succeq 0$ are weighting matrices; and the vectors $u_{\min}, u_{\max} \in \mathbb{R}^{n_u}$ define the input bounds. We use $\mathbb{E}[\cdot]$ for expectation over all random variables in the expression, unless conditioning or a subscript specifies otherwise. In \eqref{eq:smpc_problem}, the expectations are over $\theta$ and the disturbance history, with $x_0$ and the input sequence held fixed. The vectors $a_t \in \mathbb{R}^{n_x}$, $t=1,\ldots,m$, define the directions
associated with the second-order safety/risk constraints, while $t=1,\ldots,H$
selects the time instant along the prediction horizon at which each constraint is enforced.
The scalar $b_t$ specifies the corresponding offset, and $c_{t}$ specifies the allowable
second-moment bound for the constraint at time $t$. The second-moment constraints \eqref{eq:2ndorderconstraints} can provide conservative sufficient approximations of two-sided probabilistic constraints via a Chebyshev-type bound \cite{paulson2020stochastic}. Specifically, for $r_t>0$ and $\epsilon_t\in(0,1)$, choosing $c_t^2\le\epsilon_t r_t^2$ ensures $\mathbb P(|a_t^\top x_t-b_t|\ge r_t)\le\epsilon_t$ by applying Markov's inequality to the squared deviation.

The nonlinear dynamics \eqref{eq:smpc-dyn} in SMPC \eqref{eq:smpc_problem} usually result in an optimization that is both nonlinear and nonconvex.
Furthermore, in the presence of the time-varying disturbances $\{w_t\}_{t=0}^H$, evaluating the stochastic costs and constraints typically requires Monte Carlo methods or scenario-based approximations, which are prohibitive for longer horizons. Consequently, solving SMPC \eqref{eq:smpc_problem} online becomes computationally expensive, posing a major hurdle for fast real-time applications.

In deterministic nonlinear MPC, Koopman operators have been used to lift the nonlinear dynamics in \eqref{eq:smpc-dyn} to a higher-dimensional space of observables in which the lifted states evolve linearly. This allows the nonlinear MPC problem to be reformulated as a linear MPC problem. With quadratic stage costs and linear constraints, the resulting optimization is a convex quadratic program (QP). Such QPs can be solved efficiently using standard solvers, with \(O(n^3)\) complexity comparable to that of solving a linear system \cite{11431115}, and admit data-independent execution time certificates that are desirable for real-time control \cite{10683964,11240592}. When second-moment constraints are included, the lifted formulation instead yields a convex quadratically constrained quadratic program (QCQP), which can also be solved efficiently using modern conic solvers such as MOSEK and CLARABEL. On the other hand, there is a robust literature on SMPC frameworks based on PCEs, which enable uncertainty propagation into the prediction model in a structured and computationally tractable manner, thereby facilitating the evaluation of stochastic objectives and moment-based constraints \cite{mishra2024polynomial,mesbah2016stochastic}. This implicitly suggests combining these two methodologies into a single framework for SMPC of nonlinear stochastic processes. 
\section{Polynomial Chaos Theory}
This section briefly reviews polynomial chaos theory for the representation of random variables and uncertainty propagation. In \eqref{eq:nonlinear}, uncertainty enters through two distinct sources: the time-varying stochastic disturbance $w_t$ and the time-invariant parametric uncertainty $\theta$. Accordingly, the parametric Koopman matrices are expected to depend on the joint uncertainty variable $(w_t,\theta)$. In this article, we approximate this dependence using PCEs, thereby enabling the integration of Koopman operators with PCE-based linear SMPC formulations.

Let $(\Omega,\mathcal F,\mathbb P)$ be a probability space, and let $\zeta:\Omega\to\Gamma_\zeta\subseteq\mathbb R^d$ be a square-integrable random vector with law $\mu_\zeta$. Assume that $\mu_\zeta$ has finite moments of all orders, and define
\begin{equation*}
L^2(\Gamma_\zeta,\mu_\zeta)
:=
\left\{
f:\Gamma_\zeta\to\mathbb R \;\middle|\; \mathbb E\!\left[|f(\zeta)|^2\right]<\infty
\right\}.
\end{equation*}

\begin{definition}
A \textbf{polynomial chaos expansion} of a square-integrable random variable $X\in L^2(\Omega,\mathcal F,\mathbb P)$ is a representation of the form
\begin{equation}\label{eq:PCE}
X(\omega)=\sum_{\alpha\in\mathbb N_0^d} c_\alpha \Phi_\alpha(\zeta(\omega)),
\end{equation}
where $\alpha=(\alpha_1,\dots,\alpha_d)\in\mathbb N_0^d$ is a multi-index degree with total degree $|\alpha|:=\alpha_1+\cdots+\alpha_d$, $\{\Phi_\alpha\}_{\alpha\in\mathbb N_0^d}$ is a complete orthonormal basis with respect to $\mu_\zeta$, satisfying $\mathbb E\!\left[\Phi_\alpha(\zeta)\Phi_\beta(\zeta)\right]=\delta_{\alpha,\beta}$ for all $\alpha,\beta\in\mathbb N_0^d$, with $\Phi_{\mathbf 0}(\zeta)\equiv 1$, and the series in \eqref{eq:PCE} converges in $L^2(\Omega,\mathcal F,\mathbb P)$.
\end{definition}

If $\{\Phi_\alpha\}_{\alpha\in\mathbb N_0^d}$ is complete in $L^2(\Gamma_\zeta,\mu_\zeta)$, then every random variable of the form $X=f(\zeta)$, with $f\in L^2(\Gamma_\zeta,\mu_\zeta)$, admits a unique polynomial chaos expansion of the form \eqref{eq:PCE}. The coefficients are given by the projection formula
$
c_\alpha=\mathbb E\!\left[X\,\Phi_\alpha(\zeta)\right].
$
Moreover, orthonormality implies
\begin{equation}
\mathbb E[X]=c_{\mathbf 0},
\qquad
\mathrm{Var}[X]=\sum_{\alpha\neq \mathbf 0} c_\alpha^2.
\end{equation}
We impose the following conditions on the uncertainty in \eqref{eq:nonlinear}.
\begin{assumption}[Polynomial chaos property]\label{ass:pce_property}
Let $w_t:\Omega\to\mathbb W\subseteq\mathbb R^{n_w}$ be an i.i.d.\ sequence with common law $\mu_w$, and let $\theta:\Omega\to\Theta\subseteq\mathbb R^{n_\theta}$ have law $\mu_\theta$. Assume that $\{w_t\}_{t\geq 0}$ is independent of $\theta$, and that $L^2(\mathbb W,\mu_w)$ and $L^2(\Theta,\mu_\theta)$ admit complete orthonormal polynomial bases
$
\left\{
\Phi_\alpha^w(w)
\right\}_{\alpha\in\mathbb N_0^{n_w}}
$
and
$
\left\{
\Phi_\beta^\theta(\theta)
\right\}_{\beta\in\mathbb N_0^{n_\theta}},
$
respectively.
\end{assumption}

Under Assumption~\ref{ass:pce_property}, define the joint uncertainty variable $\xi_t:=(w_t,\theta)$, the joint uncertainty space $\Gamma_\xi:=\mathbb W\times\Theta$, and the corresponding probability law $\mu_\xi:=\mu_w\otimes\mu_\theta$. By tensorization,
$
\left\{
\Phi_\alpha^w(w)\Phi_\beta^\theta(\theta)
\right\}_{\alpha\in\mathbb N_0^{n_w},\,\beta\in\mathbb N_0^{n_\theta}}
$
forms a complete orthonormal basis of $L^2(\Gamma_\xi,\mu_\xi)$. We fix an ordering of the index pairs $(\alpha,\beta)$ by nondecreasing total degree $|\alpha|+|\beta|$ and denote the $k$th pair by $(\alpha(k),\beta(k))$. The corresponding product basis functions are then relabeled as
$
\Phi_k^\xi(\xi)
:=
\Phi_{\alpha(k)}^w(w)\Phi_{\beta(k)}^\theta(\theta),
\quad k\in\mathbb N_0,
$
with $\Phi_0^\xi\equiv 1$. Consequently, every square-integrable function of $\xi_t$ admits a unique polynomial chaos expansion in the basis $\{\Phi_k^\xi\}_{k\in\mathbb N_0}$. In particular, we show later that the entries of the parametric Koopman matrices belong to $L^2(\Gamma_\xi,\mu_\xi)$. Their dependence on $\xi=(w,\theta)$ therefore admits a convergent PCE and may be approximated by a finite-order truncation.

\section{Polynomial Parametric Koopman Operators for Stochastic Systems}
This section develops an extension of conventional Koopman models for stochastic systems, which we refer to as \emph{Polynomial Parametric Koopman Operators} (PPKOs). As discussed in the introduction, the uncertain nonlinear dynamics in \eqref{eq:smpc-dyn} generally cannot be efficiently represented by a single uncertainty-independent Koopman operator. This motivates a family of Koopman operators parameterized by the joint uncertainty variables $\xi=(w,\theta)$. Parametric Koopman models were introduced in \cite{guo2025learning} to construct Koopman approximations of dynamical systems with unknown parameters.

We begin by adapting the parametric Koopman formulation in \cite{guo2025learning} to autonomous systems parameterized by the joint uncertainty. Let $\xi:=(w,\theta)\in\Gamma_\xi:=\mathbb{W}\times\Theta$ denote a generic realization of the time-varying disturbance $w$ and the time-invariant uncertain parameter $\theta$. For each fixed $\xi\in\Gamma_\xi$, consider the parametrized autonomous discrete-time dynamics
\begin{equation}\label{eq:param_dyn}
x_{t+1}=f(x_t,\xi).
\end{equation}
Although the realization $w_t$ varies along a stochastic trajectory, the operator-theoretic construction below is defined for each fixed realization $\xi$ in a single step. Let $\psi:\mathbb{X}\to\mathbb{R}$ be an observable function. The dynamics \eqref{eq:param_dyn} induce the evolution $\psi(x_{t+1})=\psi(f(x_t,\xi))$, which motivates the definition of the parametrized Koopman operator $\cK(\xi)$ as
\begin{equation}\label{eq:param_koopman_def}
\cK(\xi)\psi(x):=\psi\!\left(f(x,\xi)\right).
\end{equation}
For every fixed realization $\xi\in\Gamma_\xi$, the Koopman operator acts linearly on observables; that is, for any observables $\psi_1,\psi_2$ and scalars $a,b\in\mathbb{R}$,
\begin{equation}\label{eq:param_koopman_lin}
\cK(\xi)\big(a\psi_1+b\psi_2\big)
=
a\,\cK(\xi)\psi_1+b\,\cK(\xi)\psi_2.
\end{equation}
Analogous to the deterministic case without uncertainty, for each fixed realization $\xi=(w,\theta)\in\Gamma_\xi$, the Koopman operator $\cK(\xi)$ is infinite-dimensional and must be approximated in a finite-dimensional subspace for numerical implementation. Following the approach in \cite{guo2025learning}, adapted here to the joint uncertainty variable $\xi$, let $\Psi(x):=\left[\psi_1(x),\dots,\psi_{n_\psi}(x)\right]^\top$ be the observable dictionary. The objective in parametric Koopman operators is to identify a matrix-valued function $K(\xi)\in\mathbb{R}^{n_\psi\times n_\psi}$ such that 
\begin{equation}
 K(\xi) \Psi(x) \approx \mathcal K(\xi)\Psi(x) \end{equation} 
for all $\xi\in\Gamma_\xi$, where $\Kop(\xi)$ acts on $\Psi$ component-wise. We further impose the following assumption on the observable dictionary:
\begin{assumption}
The dictionary $\Psi(x)$ contains $n_\psi\ge n_x+1$ observables, ordered so that $\psi_1(x):=1$ and $\psi_{i+1}(x):=x_i$ for $i=1,\ldots,n_x$.
\end{assumption}

\textbf{Polynomial Dependence of the Koopman Matrix}. An important distinction between the approach in \cite{guo2025learning} and the proposed framework is the parametrization of the finite-dimensional Koopman matrix. In \cite{guo2025learning}, the dependence of the Koopman matrix on the unknown parameters is represented by a deep neural network. In contrast, we approximate the dependence of $K(\xi)$ on the joint uncertainty variable $\xi=(w,\theta)$ using a PCE. This polynomial structure is particularly useful for stochastic control because the orthogonality of the PCE basis facilitates the evaluation of the expectations and second moments appearing in the SMPC problem. For a prescribed total degree $p$, let
$
N_p:=\binom{n_w+n_\theta+p}{p}
$
denote the number of retained joint polynomial basis functions. Using the ordering introduced in the preceding section, we approximate $K(\xi)$ by the truncated PCE
\begin{equation}\label{eq:PPKO}
K_p(\xi)
=
\sum_{k=0}^{N_p-1}
K_k\Phi_k^\xi(\xi),
\end{equation}
where $K_k\in\mathbb{R}^{n_\psi\times n_\psi}$ are the PCE coefficient matrices. We refer to the finite-dimensional representation in \eqref{eq:PPKO} as a \emph{polynomial parametric Koopman operator} (PPKO).

% Equivalently, after reindexing the retained joint basis functions as $\{\Phi_k^\xi(\xi)\}_{k=0}^{N_p-1}$, \eqref{eq:PPKO} may be written as \begin{equation} K_p(\xi) = \sum_{k=0}^{N_p-1} K_k\Phi_k^\xi(\xi). \end{equation} 

\begin{remark}
A priori, it is unclear whether restricting the parametric Koopman representation to a finite-order PCE in \eqref{eq:PPKO} introduces an additional approximation gap. To address this question, we establish the approximation properties of the PPKO framework and its convergence as the polynomial order is increased later in this section.
\end{remark}
Once the PPKO is identified, state predictions for a realization of the joint uncertainty sequence $\{\xi_t=(w_t,\theta)\}_{t\ge0}$ can be computed via the linear parameter-varying (LPV) dynamics in the lifted space,
\begin{equation}\label{eq:LTI}
z_{t+1} \approx K_p(\xi_t)z_{t},~x_{t+1} = Cz_{t+1} ,
\end{equation}
where $z_0=\Psi(x_0)$ and the output matrix $C\in\mathbb R^{n_x\times n_\psi}$ is
\[
C=[0_{n_x\times1}, I_{n_x}, 0_{n_x\times(n_\psi-1-n_x)}].
\]
The assumed dictionary ordering ensures $C\Psi(x)=x$.

\subsection{PPKO for Controlled Systems}
We now extend the above PPKO construction to the controlled stochastic system in \eqref{eq:nonlinear}, where the control input is treated separately from the joint uncertainty variable $\xi=(w,\theta)$.
For each fixed $u\in\mathbb U$ and $\xi\in\Gamma_\xi$, the state-transition map in \eqref{eq:nonlinear} induces the Koopman operator acting on observables $\psi:\mathbb{X}\to\mathbb{R}$,
\begin{equation}
\mathcal K(u,\xi)\psi(x)
:=
\psi\!\left(f(x,u,\xi)\right).
\end{equation}
For every fixed $(u,\xi)$, $\mathcal K(u,\xi)$ acts linearly on the observable $\psi$.
To obtain a finite-dimensional model suitable for control, we approximate the action of the Koopman operator $\mathcal K(u,\xi)$ on the observable dictionary $\Psi(x):=\left[\psi_1(x),\dots,\psi_{n_\psi}(x)\right]^\top$ with a linear parameter-varying (LPV) predictor of the form
\begin{equation}\label{eq: affine in input}
\mathcal K(u,\xi) \Psi(x) \approx A(\xi)\Psi(x) + B(\xi)u
\end{equation}
where $\Kop(u,\xi)$ acts on $\Psi$ component-wise, $A(\xi)\in\mathbb R^{n_\psi\times n_\psi}$ and $B(\xi)\in\mathbb R^{n_\psi\times n_u}$. We show in Section~\ref{sec: condensed} that the LPV structure in \eqref{eq: affine in input} enables condensed deterministic reformulations of the SMPC problem.
\begin{remark}\label{remark: structural bias}
The LPV predictor in \eqref{eq: affine in input} is a structurally biased model class rather than a universally exact controlled Koopman representation, which generally requires lifting the input as well; see \cite{iacob2024koopman,heeg2026limitations,shang2026existence}. Consequently, increasing the number of observables $n_\psi$ alone generally does not drive the prediction error arbitrarily close to zero. In return, its affine-in-input structure enables condensed deterministic SMPC formulations detailed in Section \ref{sec: condensed} with reduced online computational cost.
\end{remark}
 Following the polynomial parameterization introduced in the previous subsection, we approximate $A(\xi)$ and $B(\xi)$ by truncated PCEs:
\begin{equation}\label{eq:controlled_ppko_pce} A_p(\xi) = \sum_{k=0}^{N_p-1}A_k\Phi_k^\xi(\xi), \qquad B_p(\xi) = \sum_{k=0}^{N_p-1}B_k\Phi_k^\xi(\xi), \end{equation}

Once the PPKO matrices in \eqref{eq:controlled_ppko_pce} are identified, the evolution of the lifted state $z=\Psi(x)$ can be approximated by an LPV system
\begin{subequations}\label{eq:controlled_ppko} \begin{align} z_{t+1} &\approx A_p(\xi_t)z_t+B_p(\xi_t)u_t,\\ x_t &= Cz_t, \end{align} \end{subequations}

\subsection{Approximation Properties of PPKOs}
As noted in the previous subsection, a priori, it is not obvious whether restricting the parametric Koopman representation to a finite-order PCE in \eqref{eq:controlled_ppko_pce} introduces an additional approximation gap. This subsection characterizes the approximation error introduced by approximating the optimal LPV matrices
$A_*(\xi)$ and $B_*(\xi)$ by finite-order PCEs $A_p(\xi)$ and $B_p(\xi)$.

To characterize the optimal LPV matrices $A_*(\xi)$ and $B_*(\xi)$, let $\nu$ denote a probability distribution on $\mathbb X\times\mathbb U$ from which the state--input pairs $(x,u)$ are sampled, and define $q(x,u):=[\Psi(x)^\top,u^\top]^\top$. For each uncertainty realization $\xi\in\Gamma_\xi$, define $A_*(\xi)$ and $B_*(\xi)$ as the joint minimizers of the population one-step prediction error:
\begin{small}
\begin{equation}\label{eq:population_controlled_koopman}
\begin{aligned}
&\big(A_*(\xi),B_*(\xi)\big)\\
&:=
\operatorname*{arg\,min}_{A,B}\;
\mathbb E_{(x,u)\sim\nu}
\left[
\left\|
\Psi\!\left(f(x,u,\xi)\right)
-A\Psi(x)-Bu
\right\|_2^2
\right],
\end{aligned}
\end{equation}
\end{small}
where $A\in\mathbb R^{n_\psi\times n_\psi}$ and $B\in\mathbb R^{n_\psi\times n_u}$. To ensure that \eqref{eq:population_controlled_koopman} is well posed and admits a unique solution, we impose the following assumption.
\begin{assumption}\label{assum:observ}
The following conditions hold:
\begin{itemize}
\item Each observable is continuous, and
\[
\begin{aligned}
\mathbb E_\nu[\|q(x,u)\|_2^2]&<\infty,\\
\mathbb E_{\mu_\xi\otimes\nu}[\|\Psi(f(x,u,\xi))\|_2^2]&<\infty.
\end{aligned}
\]
\item For every compact $K\subseteq\Gamma_\xi$, there exists a nonnegative $h_K\in L^2(\mathbb X\times\mathbb U,\nu)$ such that $\|\Psi(f(x,u,\xi))\|_2\le h_K(x,u)$ for $\nu$-almost every $(x,u)$ and all $\xi\in K$.
\item The regressor Gramian
$
\Sigma:=\mathbb E_{\nu}\!\left[q(x,u)q(x,u)^\top\right]
$
is positive definite.
\end{itemize}
\end{assumption}
Under Assumption~\ref{assum:observ}, the minimizer in \eqref{eq:population_controlled_koopman} is unique and satisfies
\begin{equation}\label{eq:population_controlled_koopman_sol}
\begin{bmatrix}
A_*(\xi) & B_*(\xi)
\end{bmatrix}
=
\mathbb E_{\nu}
\left[
\Psi\!\left(f(x,u,\xi)\right)q(x,u)^\top
\right]\Sigma^{-1}.
\end{equation}
Thus, for each realization of the joint uncertainty $\xi$, $A_*(\xi)$ and $B_*(\xi)$ define the optimal LPV predictor in the mean-square sense induced by the data distribution $\nu$. Moreover, we have the following lemma:
\begin{lemma}
\label{lem:lpv_continuity}
Under Assumptions~\ref{assum:system} and~\ref{assum:observ}, the optimal LPV matrix functions
$
A_*:\Gamma_\xi\to\mathbb R^{n_\psi\times n_\psi}
$
and
$
B_*:\Gamma_\xi\to\mathbb R^{n_\psi\times n_u},
$
defined by \eqref{eq:population_controlled_koopman}, are continuous.
\end{lemma}
\begin{proof}
Let $\xi_n$ be a sequence such that $\xi_n \to \xi$. Then we have 
\begin{align*}
&\lim_{n \to \infty}
\begin{bmatrix}
A_*(\xi_n) & B_*(\xi_n)
\end{bmatrix}\\
&= \lim_{n\to \infty}
\mathbb E_{\nu}
\left[
\Psi\!\left(f(x,u,\xi_n)\right)q(x,u)^\top
\right]\Sigma^{-1} \\
&= \mathbb E_{\nu} \lim_{n\to \infty}
\left[
\Psi\!\left(f(x,u,\xi_n)\right)q(x,u)^\top
\right]\Sigma^{-1} \\
&= \mathbb E_{\nu}
\left[
\Psi\!\left(f(x,u,\xi)\right)q(x,u)^\top
\right]\Sigma^{-1} \\
&= \begin{bmatrix}
A_*(\xi) & B_*(\xi)
\end{bmatrix}
\end{align*}
where the second equality follows from the Dominated Convergence Theorem. Indeed, $K:=\{\xi\}\cup\{\xi_n:n\geq1\}$ is compact, and Assumption~\ref{assum:observ} bounds the integrand in Frobenius norm by $h_K(x,u)\|q(x,u)\|_2$, which is integrable by the Cauchy--Schwarz inequality. The third equality follows from the continuity of $\Psi \circ f$ since $f$ is continuous by Assumption \ref{assum:system} and $\Psi$ is continuous by Assumption \ref{assum:observ}.
\end{proof}
Having established the continuity of $[A_*, B_*]$ in Lemma \ref{lem:lpv_continuity}, the collection of least squares problems for each $\xi \in \Gamma_\xi$ in \eqref{eq:population_controlled_koopman} may be viewed as a single optimization problem over continuous, square-integrable matrix-valued functions. Specifically, if we define
\begin{equation}\label{eq:continuous_lpvs}
\begin{aligned}
\mathcal C:=\Big\{(A,B):\;&A\in C(\Gamma_\xi;\mathbb R^{n_\psi\times n_\psi}),\\
&B\in C(\Gamma_\xi;\mathbb R^{n_\psi\times n_u}),\\
&\mathbb E_\xi[\|A(\xi)\|_F^2+\|B(\xi)\|_F^2]<\infty\Big\},
\end{aligned}
\end{equation}
and the single-step population prediction error $\mathcal J : \mathcal{C} \to \mathbb{R}^{\geq 0}$ given by 
{\small
\begin{equation}\label{eq:integrated_population_error}
\mathcal J(A,B)
:=
\underset{\substack{\xi\sim\mu_\xi\\(x,u)\sim\nu}}{\mathbb E}
\left[
\left\|
\Psi\!\left(f(x,u,\xi)\right)
-A(\xi)\Psi(x)-B(\xi)u
\right\|_2^2
\right].
\end{equation}
}
Then the optimal LPV matrices satisfy the following functional minimization problem over the function space $\mathcal C$:
\begin{equation}\label{eq:continuous_population_minimization}
[A_*,B_*]
\in
\operatorname*{arg\,min}_{(A,B)\in\mathcal C}
\mathcal J(A,B).
\end{equation}
By \eqref{eq:population_controlled_koopman_sol} and the Cauchy--Schwarz inequality,
\[
\begin{aligned}
&\mathbb E_\xi[\|[A_*(\xi)\;B_*(\xi)]\|_F^2]\\
&\quad\le \|\Sigma^{-1}\|_2^2\,\mathbb E_\nu[\|q\|_2^2]\\
&\qquad\times\mathbb E_{\mu_\xi\otimes\nu}[\|\Psi(f(x,u,\xi))\|_2^2]
<\infty.
\end{aligned}
\]
Thus, $A_*$ and $B_*$ are square integrable with respect to $\mu_\xi$ and, by Lemma~\ref{lem:lpv_continuity}, belong to $\mathcal C$. By Assumption~\ref{ass:pce_property}, they therefore admit convergent PCE representations:
\begin{equation}\label{eq:optimal_lpv_pce}
A_*(\xi)
=
\sum_{k=0}^{\infty}A_{*,k}\Phi_k^\xi(\xi),
\qquad
B_*(\xi)
=
\sum_{k=0}^{\infty}B_{*,k}\Phi_k^\xi(\xi).
\end{equation}
This implicitly suggests that $[A_*, B_*]$ can be approximated using a truncated finite-order PCE, which justifies the PPKO representation \eqref{eq:controlled_ppko_pce}. Formally, if we define 
\begin{equation}\label{eq:finite_order_ppko_class}
\begin{aligned}
\mathcal C_p
:=
\Bigg\{
[A_p \: B_p]\in\mathcal C:\;&
A_p(\xi)=\sum_{k=0}^{N_p-1}A_k\Phi_k^\xi(\xi),\\
&
B_p(\xi)=\sum_{k=0}^{N_p-1}B_k\Phi_k^\xi(\xi)
\Bigg\}.
\end{aligned}
\end{equation}
to be the vector space of polynomial functions with total degree at most $p$, we can consider the following finite order PCE-constrained minimization problem:
\begin{equation}\label{eq:finite_order_population_minimization}
[A_{p,*} \: B_{p,*}]
\in
\operatorname*{arg\,min}_{(A,B)\in\mathcal C_p}
\mathcal J(A,B).
\end{equation}
Since $\mathcal C_p \subseteq \mathcal C$, the single step errors satisfy
\[
\mathcal J(A_{p,*},B_{p,*}) \geq \mathcal J(A_{*},B_{*}),
\]
which shows that the finite-order PCE representation in the PPKO framework incurs an additional single-step prediction error. However, this additional error can be made arbitrarily low by increasing the approximation degree $p$, which is shown by the following theorem:
\begin{theorem}
Under Assumptions~\ref{assum:system}, \ref{assum:observ}, and~\ref{ass:pce_property}, the optimal
finite-order PPKO prediction error satisfies
\begin{equation}
\lim_{p \to \infty}
\mathcal{J}\!\left(\mathcal{M}_{p,*}\right)
=
\mathcal{J}\!\left(\mathcal{M}_*\right),
\end{equation}
where
\[
\mathcal{M}_{p,*}
=
\begin{bmatrix}
A_{p,*} & B_{p,*}
\end{bmatrix},
\qquad
\mathcal{M}_*
=
\begin{bmatrix}
A_* & B_*
\end{bmatrix}.
\]
\end{theorem}

\begin{proof}
Denote
\[
\mathcal{M}
=
\begin{bmatrix}
A & B
\end{bmatrix},
\qquad
\mathcal{M}_*
=
\begin{bmatrix}
A_* & B_*
\end{bmatrix}.
\]
Then, we can rewrite \eqref{eq:integrated_population_error} as
\begin{equation}
\label{eq:integrated_population_error_M}
\mathcal{J}(\mathcal{M})
:=
\mathbb{E}_{\substack{\xi \sim \mu_\xi\\(x,u)\sim\nu}}
\left[
\left\|
\Psi\!\left(f(x,u,\xi)\right)
-
\mathcal{M}(\xi)q(x,u)
\right\|_2^2
\right],
\end{equation}
where
\[
q(x,u)
=
\begin{bmatrix}
\Psi(x)^\top & u^\top
\end{bmatrix}^{\top}.
\]
Since $\mathcal{M}_*$ is an optimal LPV matrix, the first-order optimality
condition gives, for $\mu_\xi$-almost every $\xi$,
\begin{equation}
\label{eq:M_minima}
\mathbb{E}_{(x,u)\sim\nu}
\left[
\left(
\Psi\!\left(f(x,u,\xi)\right)
-
\mathcal{M}_*(\xi)q(x,u)
\right)
q(x,u)^\top
\right]
=
0.
\end{equation}
Equations~\eqref{eq:integrated_population_error_M} and~\eqref{eq:M_minima}
then give
\begin{align}
\mathcal{J}(\mathcal{M})
&=
\mathcal{J}(\mathcal{M}_*)\nonumber\\
&\quad+
\mathbb{E}_{\substack{\xi \sim \mu_\xi\\(x,u)\sim\nu}}
\left[
\left\|
\left(
\mathcal{M}(\xi)-\mathcal{M}_*(\xi)
\right)q(x,u)
\right\|_2^2
\right],
\label{eq:objective_decomposition}
\end{align}
where the cross term vanishes by \eqref{eq:M_minima}.

Define
\begin{equation}
\label{eq:truncated_pce}
A_p(\xi)
=
\sum_{k=0}^{N_p-1}
A_{*,k}\Phi_k^\xi(\xi),
\qquad
B_p(\xi)
=
\sum_{k=0}^{N_p-1}
B_{*,k}\Phi_k^\xi(\xi),
\end{equation}
as the truncated PCE approximations of $A_*$ and $B_*$, respectively, and
let
\[
\mathcal{M}_p(\xi)
=
\begin{bmatrix}
A_p(\xi) & B_p(\xi)
\end{bmatrix}.
\]
Since the PCE basis
$\{\Phi_k^\xi\}_{k\in\mathbb{N}_0}$ is complete in
$L^2(\Gamma_\xi,\mu_\xi)$, we have
\begin{equation}
\lim_{p\to\infty}
\mathbb{E}_{\xi\sim\mu_\xi}
\left[
\left\|
\mathcal{M}_*(\xi)-\mathcal{M}_p(\xi)
\right\|_F^2
\right]
=
0.
\end{equation}
It follows from \eqref{eq:objective_decomposition} that
\begin{align}
0
&\leq
\mathcal{J}(\mathcal{M}_p)
-
\mathcal{J}(\mathcal{M}_*)
\nonumber\\
&=
\mathbb{E}_{\substack{\xi \sim \mu_\xi\\(x,u)\sim\nu}}
\left[
\left\|
\left(
\mathcal{M}_p(\xi)-\mathcal{M}_*(\xi)
\right)q(x,u)
\right\|_2^2
\right]
\nonumber\\
&\leq
\mathbb{E}_{\xi\sim\mu_\xi}
\left[
\left\|
\mathcal{M}_p(\xi)-\mathcal{M}_*(\xi)
\right\|_F^2
\right]
\mathbb{E}_{(x,u)\sim\nu}
\left[
\left\|q(x,u)\right\|_2^2
\right].
\end{align}
Therefore,
\begin{equation}
\lim_{p\to\infty}
\mathcal{J}(\mathcal{M}_p)
=
\mathcal{J}(\mathcal{M}_*).
\end{equation}

Finally, since $\mathcal{M}_{p,*}$ is optimal over the set of order-$p$
PPKOs and $\mathcal{M}_p$ is an admissible order-$p$ PPKO,
\begin{equation}
\mathcal{J}(\mathcal{M}_*)
\leq
\mathcal{J}(\mathcal{M}_{p,*})
\leq
\mathcal{J}(\mathcal{M}_p).
\end{equation}
Taking the limit as $p\to\infty$ and applying the squeeze theorem yields
\begin{equation}
\lim_{p\to\infty}
\mathcal{J}(\mathcal{M}_{p,*})
=
\mathcal{J}(\mathcal{M}_*).
\end{equation}
\end{proof}
The above theorem shows that although parameterizing the PPKO matrices as a finite-order PCE incurs an additional single-step prediction error, it can be mitigated by increasing the degree of the approximation. We further remark that, empirically, only a low truncation degree is required (between 2--5), because additional regularity properties of the observables and dynamics cause subexponential decay of the PCE coefficients of $A_*$ and $B_*$.
\subsection{PPKO Training using Extended Dynamic Mode Decomposition with Dictionary Learning}
Having established the PPKO framework and its approximation properties in the preceding subsections, we now describe how the PPKO coefficient matrices $\{A_k\}_{k=0}^{N_p-1}$ and $\{B_k\}_{k=0}^{N_p-1}$ are identified from sampled one-step transition data. The proposed training procedure also incorporates Extended Dynamic Mode Decomposition with Dictionary Learning (EDMD-DL) \cite{li2017extended} to jointly optimize the lifted observables with the PCE matrices, thereby improving prediction accuracy.

Analogously to the conventional EDMD-DL approach, we simulate the nonlinear
system~\eqref{eq:nonlinear} to obtain one-step transition data, either by independently sampling $(x,u)\sim\nu$ and
$\xi=(w,\theta)\sim\mu_\xi$, or by extracting adjacent transitions from simulated trajectories,
$\{(x_j,u_j,\xi_j),x_{j,+}\}_{j=1}^{M}$, where $M$ is the total number of
samples, $\xi_j=(w_j,\theta_j)$, and
$
x_{j,+}=f(x_j,u_j,\xi_j).
$
For trajectory data, $\theta$ remains fixed within each trajectory and $w$ is redrawn at every step; the transitions need not be independent.
We further assume that the observable dictionary $\Psi(x;w_\Psi)$ is
parameterized by neural-network weights $w_\Psi$. The PPKO coefficient
matrices $\{A_k\}_{k=0}^{N_p-1}$ and $\{B_k\}_{k=0}^{N_p-1}$, together with the
observable weights $w_\Psi$, are then identified by minimizing the empirical
least-squares loss
\begin{scriptsize}
\begin{equation}
\label{eq:ppko_control_ls}
\begin{aligned}
\min_{\substack{
    \{A_k\}_{k=0}^{N_p-1},\\
    \{B_k\}_{k=0}^{N_p-1}, \\
    w_\Psi
}}
\;
&\frac{1}{M}
\sum_{j=1}^{M}
\Bigg\|
\Psi(x_{j,+};w_\Psi)
-
\Bigg(
\sum_{k=0}^{N_p-1}
A_k\Phi_k^\xi(\xi_j)
\Bigg)
\Psi(x_j;w_\Psi)
\\
&\hspace{5.2em}
-
\Bigg(
\sum_{k=0}^{N_p-1}
B_k\Phi_k^\xi(\xi_j)
\Bigg)u_j
\Bigg\|_2^2
\end{aligned}
\end{equation}
\end{scriptsize}
Under independent sampling from $\nu\otimes\mu_\xi$ and for a fixed $w_\Psi$, the loss function in
\eqref{eq:ppko_control_ls} is the empirical counterpart of the population
single-step prediction error $\mathcal{J}$ defined in
\eqref{eq:integrated_population_error}. For trajectory-derived data, the loss instead is the empirical counterpart of the expected single-step prediction error under the joint sampling distribution induced by the trajectories. We further note that for a fixed $w_\Psi$, the least-squares problem \eqref{eq:ppko_control_ls} is linear and convex in $\{A_k\}_{k=0}^{N_p-1}, \{B_k\}_{k=0}^{N_p-1}$ and can be solved at the end of every epoch. After optimizing $w_\Psi$ with fixed $\{A_k\}_{k=0}^{N_p-1}, \{B_k\}_{k=0}^{N_p-1}$, least-squares regression is performed to determine the PCE coefficients $\{A_k\}_{k=0}^{N_p-1}, \{B_k\}_{k=0}^{N_p-1}$. This alternating minimization method converges rapidly in practice \cite{li2017extended}. The loss in \eqref{eq:ppko_control_ls} can also be adapted to include multi-step prediction errors to improve prediction accuracy over the SMPC horizon. Once the PPKO coefficient matrices
$\{A_k\}_{k=0}^{N_p-1}$ and $\{B_k\}_{k=0}^{N_p-1}$, together with the observable
weights $w_\Psi$, have been identified for the PPKO model
\eqref{eq:controlled_ppko_pce}, the resulting model can be used to propagate
the nonlinear system \eqref{eq:nonlinear} in lifted coordinates according to
\eqref{eq:controlled_ppko}. This enables the condensed deterministic
reformulation of the SMPC problem \eqref{eq:smpc_problem}, presented in the
next section.
\section{Condensed Deterministic SMPC Reformulation}\label{sec: condensed}

Once a PCE LPV model such as \eqref{eq:controlled_ppko} is identified, the conventional approach for deterministic SMPC reformulation is to apply Stochastic Galerkin Projections (SGPs) to \eqref{eq:smpc_problem}, which results in a deterministic expanded system for the PCE coefficients. The conventional SGP approach increases the state dimension proportionately with the number of retained PCE polynomials, which increases the computational burden of the deterministic reformulation of SMPC. SGPs also induce additional truncation error during uncertainty propagation and the errors may propagate along the prediction horizon. In this section, inspired by condensing methods in MPC, we propose a deterministic reformulation approach that explicitly avoids the use of SGPs, which eliminates truncation errors and yields an online optimization whose decision vector is the stacked input sequence, with dimension depending only on the control horizon and input dimension.
\subsection{PPKO-SMPC}
Consider the nonlinear SMPC problem given by \eqref{eq:smpc_problem}. By approximating the nonlinear dynamics \eqref{eq:smpc-dyn} with the PPKO model \eqref{eq:controlled_ppko}, we obtain the following PPKO-SMPC problem:

\noindent
\begin{minipage}{\columnwidth}
\noindent\rule{\columnwidth}{0.4pt}
\vspace{-0.7em}
\textbf{PPKO-SMPC}
\vspace{-0.4em}

{\small
\setlength{\abovedisplayskip}{2pt}
\setlength{\belowdisplayskip}{2pt}
\setlength{\jot}{2pt}
\begin{subequations}\label{eq:PPKO_smpc_problem}
\begin{align}
\min_{\{u_t\}_{t=0}^{H-1}} 
&~ \sum_{t=0}^{H-1} 
\mathbb{E}\big[x_t^\top Q x_t + u_t^\top R u_t \big]
+ \mathbb{E}\big[x_H^\top Q_f x_H\big]
\label{eq:ppko-smpc-cost}
\\
\text{s.t.}~
& z_0 = \Psi(x_0),
\label{eq:ppko-smpc-lifting}
\\
& z_{t+1} = A_p(\xi_t)z_t + B_p(\xi_t) u_t,
\quad t=0,\dots,H-1,
\label{eq:ppko-smpc-dyn}
\\
& x_t = C z_t,
\quad t=0,\dots,H,
\label{eq:ppko-smpc-output}
\\
& x_t^{\min} \le \mathbb{E}[x_t] \le x_t^{\max},
\quad t=1,\dots,H,
\label{eq:ppko-smpc-mean}
\\
& \mathbb{E}\big[(a_t^\top x_t-b_t)^2\big] \le c_t^2,
\quad t=1,\dots,H,
\label{eq:ppko-2ndorderconstraints}
\\
& u_{\min} \le u_t \le u_{\max},
\quad t=0,\dots,H-1.
\label{eq:ppko-smpc-input}
\end{align}
\end{subequations}
}
\par\noindent\rule{\columnwidth}{0.4pt}\par
\end{minipage}

Conventional approaches to stochastic programs with expectation constraints, such as \eqref{eq:PPKO_smpc_problem}, rely on sampling-based approximations, including sample-average and scenario methods, or explicit propagation of state moments using methods such as SGPs. Their online computational burden can grow substantially with the dimension of the observable dictionary $\Psi$ and the number of PCE basis functions, which itself suffers from the curse of dimensionality. The following theorem shows that \eqref{eq:PPKO_smpc_problem} can instead be condensed into a deterministic optimization problem in the stacked input vector by precomputing a set of moment matrices offline.
\begin{theorem}
\label{thm:condensed_ppko_smpc}
Consider the PPKO-SMPC problem \eqref{eq:PPKO_smpc_problem} under the
affine-in-control PPKO approximation \eqref{eq:ppko-smpc-dyn}. Define
the lifted initial condition
$
z_0:=\Psi(x_0)
$
and the stacked input vector
$
U
:=
[
u_0^\top ,\cdots , u_{H-1}^\top
]^{\top}
$,  
which is the only online decision vector. Let the offline condensed
matrices $\mathbf{S}$, $\bar{\mathbf{G}}_t$, $\mathbf{M}_t$, and
$\mathbf{c}_t$ be partitioned with respect to
$[z_0^\top,U^\top]^\top$ as
\[
\mathbf{S}
=
\begin{bmatrix}
\mathbf{S}_{zz} & \mathbf{S}_{zU}\\
\mathbf{S}_{Uz} & \mathbf{S}_{UU}
\end{bmatrix},
\quad
\bar{\mathbf{G}}_t
=
\begin{bmatrix}
\bar{\mathbf{E}}_t & \bar{\mathbf{F}}_t
\end{bmatrix},
\]
\[
\mathbf{M}_t
=
\begin{bmatrix}
\mathbf{M}_{zz,t} & \mathbf{M}_{zU,t}\\
\mathbf{M}_{Uz,t} & \mathbf{M}_{UU,t}
\end{bmatrix},
\qquad
\mathbf{c}_t
=
\begin{bmatrix}
\mathbf{c}_{z,t}\\
\mathbf{c}_{U,t}
\end{bmatrix}.
\]
Define the initial-state-dependent vectors and scalars
\[
\mathbf{s}(x_0):=\mathbf{S}_{Uz}z_0,\qquad
\mathbf{m}_t(x_0):=\mathbf{M}_{Uz,t}z_0-b_t\mathbf{c}_{U,t},
\]
\[
d_t(x_0):=
z_0^\top\mathbf{M}_{zz,t}z_0
-2b_t\mathbf{c}_{z,t}^\top z_0
+b_t^2.
\]
Then \eqref{eq:PPKO_smpc_problem} is equivalent, up to an
input-independent cost term, to the following deterministic QCQP in \(U\):

\noindent
\begin{minipage}{\columnwidth}
\normalfont
\noindent\rule{\columnwidth}{0.4pt}
\vspace{-0.7em}
\textbf{Condensed PPKO-SMPC}
\vspace{-0.4em}

{\scriptsize
\begin{subequations}
\label{eq:det_smpc_problem}
\begin{align}
\min_{U}\quad
&
U^\top \mathbf{S}_{UU}U+2\mathbf{s}(x_0)^\top U
\label{eq:det-cost}
\\
\mathrm{s.t.}\quad
&
x_t^{\min}
\leq
\bar{\mathbf{E}}_t z_0+\bar{\mathbf{F}}_tU
\leq
x_t^{\max},
\quad
t=1,\ldots,H,
\label{eq:det-1stmoment}
\\
&
U^\top\mathbf{M}_{UU,t}U
+2\mathbf{m}_t(x_0)^\top U
+d_t(x_0)
\leq
c_t^2,
\quad
t=1,\ldots,H,
\label{eq:det-2ndmoment}
\\
&
\mathbf{1}_H\otimes u_{\min}
\leq
U
\leq
\mathbf{1}_H\otimes u_{\max}.
\label{eq:stacked-input-bounds}
\end{align}
\end{subequations}
}

\noindent\rule{\columnwidth}{0.4pt}
\end{minipage}

where the condensed matrices
$\{\mathbf{S},\bar{\mathbf{G}}_t,\mathbf{M}_t,\mathbf{c}_t\}$
are explicitly constructed in the proof and are evaluated online only
through the fixed parameter \(z_0=\Psi(x_0)\) and the decision vector \(U\).
\end{theorem}

\begin{proof}
First define the stacked vectors
\[
\begin{aligned}
X&=[x_1^\top,\cdots{},x_H^\top]^\top,\\
U&=[u_0^\top,\cdots{},u_{H-1}^\top]^\top.
\end{aligned}
\]
For notational convenience in constructing the offline matrices, introduce
the auxiliary vector
\[
\eta:=
\begin{bmatrix}
z_0^\top & U^\top
\end{bmatrix}^{\top},
\]
where \(z_0=\Psi(x_0)\) is fixed during the online optimization, and let
\(\mathbf{P}_U=[0\ I_{Hn_u}]\).
Moreover, let
$
\mathbf{w}:=(w_0,\ldots,w_{H-1})
$
denote the disturbance sequence, where the $w_t$ are i.i.d., and recall
that $\xi_t=(w_t,\theta)$, with $\theta$ fixed over the prediction
horizon. Recursively applying \eqref{eq:ppko-smpc-dyn}, and noting that
$z_0=\Psi(x_0)$, results in
\begin{equation}
X
=
\mathbf{E}(\theta,\mathbf{w})z_0
+
\mathbf{F}(\theta,\mathbf{w})U
=
\mathbf{G}(\theta,\mathbf{w})\eta,
\end{equation}
where
\[
\mathbf{G}(\theta,\mathbf{w})
:=
\begin{bmatrix}
\mathbf{E}(\theta,\mathbf{w})
&
\mathbf{F}(\theta,\mathbf{w})
\end{bmatrix},
\]
and
{\scriptsize
\setlength{\arraycolsep}{1.25pt}
\[
\mathbf{E}(\theta,\mathbf{w}) =
\!\begin{bmatrix}
C A_p(\xi_0) \\
C A_p(\xi_1)A_p(\xi_0) \\
\vdots \\
C A_p(\xi_{H-1})\cdots A_p(\xi_0)
\end{bmatrix}\!,
\]
\begin{equation}
\label{eq:E_F_def}
\begin{gathered}
\mathbf{F}(\theta,\mathbf{w}) ={}\\
\begin{bmatrix}
C B_p(\xi_0) & 0 & \cdots & 0 \\
C A_p(\xi_1) B_p(\xi_0) & C B_p(\xi_1) & \cdots & 0 \\
\vdots & \vdots & \ddots & \vdots \\
C \Pi_{k=1}^{H-1} A_p(\xi_k) B_p(\xi_0)
&
C \Pi_{k=2}^{H-1} A_p(\xi_k)B_p(\xi_1)
&
\cdots
&
C B_p(\xi_{H-1})
\end{bmatrix}.
\end{gathered}
\end{equation}
}

Now the SMPC LQR cost \eqref{eq:ppko-smpc-cost} can be rewritten,
up to the constant term $x_0^\top Qx_0$, as
\begin{equation}
\label{eq:pcecost}
\begin{aligned}
J(\eta)
&=
\mathbb{E}
\Big[
(\mathbf{P}_U\eta)^\top
\bar{R}
(\mathbf{P}_U\eta)
\\
&\qquad\qquad
+
\big(
\mathbf{G}(\theta,\mathbf{w})\eta
\big)^\top
\bar{Q}
\big(
\mathbf{G}(\theta,\mathbf{w})\eta
\big)
\Big],
\end{aligned}
\end{equation}
where
\[
\begin{aligned}
\bar{R}&=\mathrm{blkdiag}(R,\cdots{},R),\\
\bar{Q}&=\mathrm{blkdiag}(Q,\cdots{},Q,Q_f).
\end{aligned}
\]
This expression can be further condensed into
\begin{equation}
J(\eta)
=
\eta^\top\mathbf{S}\eta,
\end{equation}
where
\[
\mathbf{S}
:=
\mathbf{P}_U^\top\bar{R}\mathbf{P}_U
+
\mathbb{E}
\!\left[
\mathbf{G}(\theta,\mathbf{w})^\top
\bar{Q}
\mathbf{G}(\theta,\mathbf{w})
\right],
\]
Partitioning $\mathbf{S}$ conformably with $[z_0^\top,U^\top]^\top$
gives
\[
J(U)
=
U^\top\mathbf{S}_{UU}U
+2(\mathbf{S}_{Uz}z_0)^\top U
+z_0^\top\mathbf{S}_{zz}z_0.
\]
The final term is independent of \(U\), and dropping it gives
\eqref{eq:det-cost}.

Furthermore, defining
$\mathbf{E}_t(\theta,\mathbf{w})$ and
$\mathbf{F}_t(\theta,\mathbf{w})$ as the $t$th block rows of
$\mathbf{E}(\theta,\mathbf{w})$ and
$\mathbf{F}(\theta,\mathbf{w})$, respectively, and defining
\begin{equation}\label{eq: Gt def}
\mathbf{G}_t(\theta,\mathbf{w})
:=
\begin{bmatrix}
\mathbf{E}_t(\theta,\mathbf{w})
&
\mathbf{F}_t(\theta,\mathbf{w})
\end{bmatrix},
\end{equation}
we have
\begin{equation}
\begin{aligned}
\mathbb{E}[x_t]
&=
\mathbb{E}
\!\left[
\mathbf{G}_t(\theta,\mathbf{w})\eta
\right]
\\
&=
\bar{\mathbf{G}}_t\eta,
\end{aligned}
\end{equation}
where
\[
\bar{\mathbf{G}}_t
:=
\mathbb{E}
\!\left[
\mathbf{G}_t(\theta,\mathbf{w})
\right]
=
\begin{bmatrix}
\bar{\mathbf{E}}_t
&
\bar{\mathbf{F}}_t
\end{bmatrix},
\]
with
\[
\bar{\mathbf{E}}_t
:=
\mathbb{E}
\!\left[
\mathbf{E}_t(\theta,\mathbf{w})
\right],
\qquad
\bar{\mathbf{F}}_t
:=
\mathbb{E}
\!\left[
\mathbf{F}_t(\theta,\mathbf{w})
\right].
\]
Substituting the fixed initial condition gives
\[
\mathbb{E}[x_t]
=
\bar{\mathbf{E}}_t z_0+\bar{\mathbf{F}}_tU,
\]
which gives \eqref{eq:det-1stmoment}.

For the second-moment constraints, since
$x_t=\mathbf{G}_t(\theta,\mathbf{w})\eta$, we have
\begin{equation}
\mathbb{E}
\!\left[
(a_t^\top x_t-b_t)^2
\right]
=
\eta^\top
\mathbf{M}_t
\eta
-
2b_t\mathbf{c}_t^\top\eta
+
b_t^2,
\end{equation}
where
\begin{align}
\label{eq:2ndmoment}
\mathbf{M}_t
&:=
\mathbb{E}
\!\left[
\mathbf{G}_t(\theta,\mathbf{w})^\top
a_ta_t^\top
\mathbf{G}_t(\theta,\mathbf{w})
\right],
\\
\mathbf{c}_t
&:=
\mathbb{E}
\!\left[
\mathbf{G}_t(\theta,\mathbf{w})^\top
a_t
\right].
\end{align}
Using the block partitions of $\mathbf{M}_t$ and $\mathbf{c}_t$, and
substituting the fixed \(z_0\), this expression becomes
\[
\begin{aligned}
&\mathbb{E}
\!\left[
(a_t^\top x_t-b_t)^2
\right]\\
={}&
U^\top\mathbf{M}_{UU,t}U
+2(\mathbf{M}_{Uz,t}z_0-b_t\mathbf{c}_{U,t})^\top U
\\
&+
z_0^\top\mathbf{M}_{zz,t}z_0
-2b_t\mathbf{c}_{z,t}^\top z_0
+b_t^2.
\end{aligned}
\]
Substituting this expression into \eqref{eq:ppko-2ndorderconstraints}
gives \eqref{eq:det-2ndmoment}.

Finally, the stacked input bounds \eqref{eq:stacked-input-bounds} are
equivalent to \eqref{eq:ppko-smpc-input}. Thus, after substituting the fixed
lifted initial condition \(z_0=\Psi(x_0)\), the only online decision
vector is \(U\). This concludes the reformulation of the PPKO-SMPC
problem \eqref{eq:PPKO_smpc_problem} as the deterministic QCQP
\eqref{eq:det_smpc_problem}.
\end{proof}
The above theorem shows that, by exploiting the linear structure of the
PPKO model and introducing the condensed matrices, the stochastic
PPKO-SMPC problem can be reformulated as the deterministic problem
\eqref{eq:det_smpc_problem}, whose decision-variable dimension is
$Hn_u$. Hence, the decision dimension depends only on the prediction
horizon and input dimension and is independent of both the number of
PCE basis functions and the number of Koopman observables used in the
PPKO model. 

\subsection{Computing the Condensed Matrices}
Although the resulting optimization problem is deterministic, the
remaining computational challenge lies in evaluating the expectations
over the parametric uncertainty and stochastic disturbances
$(\theta,\mathbf{w})$ required to construct the condensed matrices
$\{\mathbf{S},\bar{\mathbf{G}}_t,\mathbf{M}_t,\mathbf{c}_t\}$. A straightforward approach is to
approximate these expectations by sampling many realizations of
$(\theta,\mathbf{w})$. However, since the dimension of the disturbance
trajectory $\mathbf{w}=(w_0,\ldots,w_{H-1})$ grows linearly with the
prediction horizon $H$, this approach can become computationally
burdensome for long horizons. To address this challenge, we propose an
offline recursive moment-condensation algorithm that uses Gaussian
quadrature to integrate over the parametric uncertainty $\theta$ and
exploits both the IID structure of $\{w_t\}$ and the PCE
structure of the PPKO model to efficiently compute the required
condensed matrices. In particular, conditional on \(\theta\), the current realization \(w_t\) is independent of the past and hence independent of \(z_t\), which allows the required conditional expectations to be propagated recursively using averaged one-step operators rather than by explicitly parameterizing the full disturbance path.

\textbf{Summary of Recursive Algorithm :} To provide some intuition on the recursive algorithm to compute the condensed matrices, we provide an overview of the algorithm:

\begin{enumerate}
\item \textbf{Quadrature Construction and Initialization:} Choose an exact Gaussian quadrature rule for the parametric uncertainty $\theta$, with nodes and weights $\left\{ \left(\theta^{(r)},\omega_r\right) \right\}_{r=1}^{N_q}, $ and fix a node $\theta=\theta^{(r)}$.

\item \textbf{Conditional First Moment:} Let $\eta :=
\begin{bmatrix}
z_0^\top & U^\top
\end{bmatrix}^{\top}$ and $\bGamma_t(w_{0:t-1}, \theta)$ be a sequence of random matrices such that for all $t \geq 0$, the matrix $\bGamma_t$ satisfies $z_t = \bGamma_t \eta $.\footnote{We omit the dependence on $(w_{0:t-1}, \theta)$ to simplify notation.} The algorithm computes $\bar \bGamma_t^{(r)} := \mathbb{E}[\bGamma_t \mid\theta = \theta^{(r)}]$ recursively, which enables the computation of the conditional condensed matrix for the expectation constraint \eqref{eq:ppko-smpc-mean} using the following identity:
$$\mathbb{E}[x_t \mid\theta = \theta^{(r)}] = C \bar \bGamma_t^{(r)} \eta.  $$

\item \textbf{Conditional Second Moment:} For any symmetric matrix $S\in\mathbb{S}^{n_z}$, define the conditional quadratic kernel 
\[
\mathcal{M}_t^{(r)}(S) := \mathbb{E} \left[ \boldsymbol{\Gamma}_t^\top S \boldsymbol{\Gamma}_t \,\middle|\, \theta=\theta^{(r)} \right].
\]
Then, for any symmetric matrix $S \in \mathbb{S}^{n_z}$, one has
\begin{equation}\label{eq: Mt def}
\mathbb{E}[z_t^T S z_t \mid\theta = \theta^{(r)}] = \eta^\top \mathcal{M}_t^{(r)}(S) \eta.
\end{equation}
The algorithm then computes the conditional condensed matrices for the quadratic costs \eqref{eq:ppko-smpc-cost} and second-moment constraints \eqref{eq:ppko-2ndorderconstraints} using the following identities:
$$\mathbb{E}[x_t^\top Q x_t \mid\theta = \theta^{(r)}] = \eta^\top \mathcal{M}_t^{(r)}(C^\top Q C) \eta.  $$
$$\mathbb{E}[x_t^\top a a^\top x_t \mid\theta = \theta^{(r)}] = \eta^\top \mathcal{M}_t^{(r)}(C^\top a a^\top C) \eta.  $$

\item \textbf{Quadrature Integration:} We then sum over the quadrature nodes to integrate the conditional moments over $\theta$. Define \[ \bar{\mathbf{G}}_t := \sum_{r=1}^{N_q} \omega_r C\bar{\bGamma}_t^{(r)}, \qquad \overline{\mathcal{M}}_t(S) := \sum_{r=1}^{N_q} \omega_r \mathcal{M}_t^{(r)}(S). \] Since the Gaussian quadrature rule is exact for the corresponding polynomial integrands, these quantities satisfy, for any symmetric matrix $S \in \mathbb{S}^{n_x}$,  \[ \mathbb{E}[x_t] = \bar{\mathbf{G}}_t\eta, \qquad \mathbb{E}[x_t^\top Sx_t] = \eta^\top \overline{\mathcal{M}}_t(C^\top SC) \eta. \]

\item \textbf{Construct Condensed Matrices:} The identities above allow us to construct the required condensed matrices. Specifically, for the second-moment constraints, define \[ \begin{aligned} \mathbf{M}_t &:= \overline{\mathcal{M}}_t \left( C^\top a_ta_t^\top C \right),\\ \mathbf{c}_t &:= \sum_{r=1}^{N_q} \omega_r \left( \bar{\bGamma}_t^{(r)} \right)^\top C^\top a_t. \end{aligned} \] It follows that \[ \mathbb{E} \left[ (a_t^\top x_t-b_t)^2 \right] = \eta^\top\mathbf{M}_t\eta - 2b_t\mathbf{c}_t^\top\eta + b_t^2. \] Similarly, up to the constant initial-state term, the condensed cost matrix is \[ \mathbf{S} := \mathbf{P}_U^\top\bar{R}\mathbf{P}_U + \sum_{t=1}^{H-1} \overline{\mathcal{M}}_t(C^\top QC) + \overline{\mathcal{M}}_H(C^\top Q_fC). \] This gives the condensed matrices $\{\mathbf{S},\bar{\mathbf{G}}_t,\mathbf{M}_t,\mathbf{c}_t\}$ appearing in \eqref{eq:det_smpc_problem}.\end{enumerate}

The remainder of this subsection is organized as follows. First, we provide analytical results on the PCE structure of the PPKO model. Second, we show how the conditional first- and second-moment matrices $\bGamma_t$ and $\mathcal{M}_t^{(r)}(S)$ are
computed recursively. Finally, we present the recursive algorithm for computing the condensed matrices and prove its correctness.

\textbf{Polynomial Chaos Decomposition:} 
We first develop some analytical results to show how the PCE structure of the PPKO model can be utilized. Note that we can write the PPKO matrices in the separable form 
\begin{equation}
\begin{aligned}
A_p(w,\theta)
&=
\sum_{| \alpha | + |\beta | \leq p }
A_{\alpha\beta}\,
\Phi_\alpha^w (w)\Phi_\beta^\theta (\theta),
\\
B_p(w,\theta)
&=
\sum_{| \alpha | + |\beta | \leq p }
B_{\alpha\beta}\,
\Phi_\alpha^w (w)\Phi_\beta^\theta (\theta),
\end{aligned}
\label{eq:AB_joint_pce}
\end{equation}
where $\{\Phi_\alpha^w \}$ and $\{\Phi_\beta^\theta \}$ are the orthonormal bases defined in Assumption \ref{ass:pce_property}. To separate the dependence on the parametric uncertainty and disturbances, by collecting the terms in the joint total-degree-$p$ PCE that share
the same basis function in $w$, the PPKO matrices admit the
representations
\begin{equation}
\label{eq:AB_w_modes}
A_p(w,\theta)
=
\sum_{|\alpha| \leq p }
A_\alpha(\theta)\Phi_\alpha^w (w),
\end{equation}
\begin{equation}
B_p(w,\theta)
=
\sum_{|\alpha| \leq p }
B_\alpha(\theta)\Phi_\alpha^w (w),
\end{equation}
where we define the mode matrices 
\begin{small}\begin{equation}\label{eq:PCE modes}
A_\alpha(\theta) :=  \sum_{|\beta| \leq p - |\alpha| } A_{\alpha\beta} \Phi_\beta^\theta (\theta), \;
B_\alpha(\theta) :=  \sum_{|\beta| \leq p - |\alpha| } B_{\alpha\beta} \Phi_\beta^\theta (\theta)
\end{equation}\end{small}
Here, $\mathbb{E}_w$ averages over a single disturbance $w\sim\mu_w$, with all other arguments held fixed. Conditional state moments are written as $\mathbb{E}[\cdot\mid\theta]$ and average over the disturbance history. By orthonormality of the disturbance basis
$\{\Phi_\alpha^w\}_{\alpha\in\mathbb{N}_0^{n_w}}$,
\begin{equation}
\label{eq:w-basis-orthonormality}
\mathbb{E}_w
\left[
\Phi_\alpha^w(w_t)
\right]
=
\delta_{\alpha,\mathbf{0}},
\qquad
\mathbb{E}_w
\left[
\Phi_\alpha^w(w_t)\Phi_\gamma^w(w_t)
\right]
=
\delta_{\alpha,\gamma},
\end{equation}
where $\mathbf{0}$ denotes the zero multi-index and
$\delta_{\alpha,\gamma}$ is the Kronecker delta for multi-indices. 

\textbf{Conditional First Moments:}
For a fixed realization $\theta = \theta^{(r)}$, to recursively compute the conditional first-moment matrices
$\bGamma_t^{(r)}$, we first define
\begin{align}
\bar A(\theta)
&:=
\mathbb{E}_w[A_p(w,\theta)]
= A_{\mathbf{0}}(\theta),\\
\bar B(\theta)
&:=
\mathbb{E}_w[B_p(w,\theta)]
= B_{\mathbf{0}}(\theta),
\end{align}
where $A_{\mathbf{0}}$ and $B_{\mathbf{0}}$ are defined in
\eqref{eq:PCE modes}. The second equality in each line follows from
\eqref{eq:w-basis-orthonormality}. We then have the following recursion.

\begin{prop}\label{prop: first moment}
The conditional first-moment matrices satisfy
\begin{equation}\label{eq:Gammat_rec}
\begin{aligned}
\bGamma_0^{(r)}
&=
\begin{bmatrix}
I_{n_z} & 0
\end{bmatrix},
\\
\bGamma_{t+1}^{(r)}
&=
A_{\mathbf{0}}(\theta^{(r)})\bGamma_t^{(r)}
+
B_{\mathbf{0}}(\theta^{(r)})\mathbf{P}_{u_t},
\end{aligned}
\end{equation}
where $\mathbf{P}_{u_t}$ is the projection matrix satisfying
$u_t=\mathbf{P}_{u_t}\eta$.
\end{prop}

\begin{proof}
The base case follows directly from the definition of $\eta$. For the
recursion, we first note that
\begin{align}
\mathbb{E}[z_{t+1}\mid\theta]
&=
\mathbb{E}
\left[
A_p(w_t,\theta)z_t+B_p(w_t,\theta)u_t
\mid\theta
\right] \\
&=
A_{\mathbf{0}}(\theta)\mathbb{E}[z_t\mid\theta]
+
B_{\mathbf{0}}(\theta)u_t.
\label{eq: conditional mean}
\end{align}
Importantly, the second equality holds because $z_t$ depends only on the
past disturbances $(w_0,\ldots,w_{t-1})$, while the i.i.d. assumption
implies that the current disturbance $w_t$ is conditionally independent
of $z_t$ given $\theta$. By definition,
\[
\mathbb{E}[z_t\mid\theta=\theta^{(r)}]
=
\bGamma_t^{(r)}\eta.
\]
It therefore follows from \eqref{eq: conditional mean} that
\[
\bGamma_{t+1}^{(r)}\eta
=
A_{\mathbf{0}}(\theta^{(r)})\bGamma_t^{(r)}\eta
+
B_{\mathbf{0}}(\theta^{(r)})\mathbf{P}_{u_t}\eta,
\]
which proves the recursive identity \eqref{eq:Gammat_rec}.
\end{proof}

Having obtained a closed recursion for the conditional first moment, we
next derive the corresponding conditional second-moment recursion, which must account for contributions from all disturbance PCE modes rather than
only the $\mathbf{0}$'th mode.

\textbf{Conditional Second Moments:} Similarly, for a fixed realization $\theta = \theta^{(r)}$, the recursion to compute the conditional quadratic kernels $\mathcal{M}_t^{(r)}(S)$ is given by the following:
\begin{prop}\label{prop: second moments}
For a fixed quadrature node $\theta=\theta^{(r)}$ and any
$S\in\mathbb{S}^{n_z}$, the conditional quadratic kernels satisfy 
\begin{align}
\mathcal{M}_0^{(r)}(S) &= 
\begin{bmatrix}
I_{n_z} & 0
\end{bmatrix}^{\!\top}
S
\begin{bmatrix}
I_{n_z} & 0
\end{bmatrix},
\label{eq:M0_def} \\
\mathcal{M}_{t+1}^{(r)}(S)
&=
\mathcal{M}_t^{(r)}\!\left(\mathcal Q^{(r)}(S)\right)
+
{\bGamma_t^{(r)}}^\top \mathcal N^{(r)}(S) \mathbf{P}_{u_t}  \nonumber \\
&\quad +
\mathbf{P}_{u_t}^\top \mathcal N^{(r)}(S)^\top \bGamma_t^{(r)}
+
\mathbf{P}_{u_t}^\top \mathcal R^{(r)}(S) \mathbf{P}_{u_t} .
\label{eq:Mt_recursion}
\end{align}
where $\bGamma_t^{(r)}$ and $\mathbf{P}_{u_t}$ are defined in Proposition~\ref{prop: first moment}, and the operators $\mathcal Q^{(r)}, \mathcal N^{(r)}, \mathcal R^{(r)}$ are defined by
\begin{align}
\mathcal Q^{(r)}(S)
&:=
\sum_{|\alpha|\leq p } A_\alpha(\theta^{(r)})^\top S A_\alpha(\theta^{(r)}),
\label{eq:Qtheta_def}
\\
\mathcal N^{(r)}(S)
&:=
\sum_{|\alpha|\leq p } A_\alpha(\theta^{(r)})^\top S B_\alpha(\theta^{(r)}),
\label{eq:Ntheta_def}
\\
\mathcal R^{(r)}(S)
&:=
\sum_{|\alpha|\leq p } B_\alpha(\theta^{(r)})^\top S B_\alpha(\theta^{(r)}).
\label{eq:Rtheta_def}
\end{align}
and $A_\alpha, B_\alpha$ are the mode matrices defined in \eqref{eq:PCE modes}.
\end{prop}

\begin{proof}
The base case \eqref{eq:M0_def} follows readily from the definition \eqref{eq: Mt def} and the following:
\[
\mathbb{E}[z_0^\top S z_0 \mid\theta = \theta^{(r)}]= \eta^\top \begin{bmatrix}
I_{n_z} & 0
\end{bmatrix}^{\!\top}
S
\begin{bmatrix}
I_{n_z} & 0
\end{bmatrix} \eta
\]
For the recursion, for each quadrature node $\theta^{(r)}$, we have
\begin{scriptsize}
\begin{align}
&\mathbb{E}
\left[
z_{t+1}^\top S z_{t+1}
\,\middle|\,
\theta=\theta^{(r)}
\right]
\nonumber\\
&=
\mathbb{E}
\left[
z_t^\top A_p(w_t,\theta)^\top S A_p(w_t,\theta) z_t
\,\middle|\,
\theta=\theta^{(r)}
\right]
\nonumber\\
&\quad+
\mathbb{E}
\left[
\begin{aligned}
&z_t^\top A_p(w_t,\theta)^\top S B_p(w_t,\theta) u_t\\
&\quad+
u_t^\top B_p(w_t,\theta)^\top S A_p(w_t,\theta)z_t
\end{aligned}
\,\middle|\,
\theta=\theta^{(r)}
\right]
\nonumber\\
&\quad+
u_t^\top
\mathbb{E}_w
\left[
B_p(w,\theta^{(r)})^\top S B_p(w,\theta^{(r)})
\right]
u_t.
\label{eq:quadratic_moment_expansion}
\end{align}
\end{scriptsize}

For the first term in \eqref{eq:quadratic_moment_expansion}, because $z_t$ depends only on the
past disturbances $(w_0,\ldots,w_{t-1})$, the i.i.d. assumption
implies that $A_p(w_t,\theta)$ is conditionally independent
of $z_t$ given $\theta$. Thus, by using the definition of $\mathcal{M}_t$ in \eqref{eq: Mt def}, the first term is equal to
\begin{equation}
\eta ^\top \mathcal M _t^{(r)} \left(\mathbb{E}_w[A_p(w,\theta^{(r)})^\top S A_p(w,\theta^{(r)})]\right) \eta
\end{equation}
Now using the mode decomposition of $A_p(w_t,\theta)$ given in \eqref{eq:PCE modes} and the orthonormality identity \eqref{eq:w-basis-orthonormality}, we obtain
\begin{small}
\begin{align}
& \mathbb{E}_w[A_p(w,\theta^{(r)})^\top S A_p(w,\theta^{(r)})] \nonumber \\
&=   \sum_{|\alpha|\leq p } A_\alpha(\theta^{(r)})^\top S A_\alpha(\theta^{(r)}) = \mathcal{Q}^{(r)}(S)\nonumber
\end{align}
\end{small}
So the first term in \eqref{eq:quadratic_moment_expansion} is equal to $\eta^\top \mathcal{M}_t^{(r)}\!\left(\mathcal Q^{(r)}(S)\right) \eta$.
A similar argument shows that the third term in \eqref{eq:quadratic_moment_expansion} is
\begin{align}
&u_t^\top
\mathbb{E}_w
\left[
B_p(w,\theta^{(r)})^\top S B_p(w,\theta^{(r)})
\right]
u_t
\nonumber\\
&=
u_t^\top \mathcal R^{(r)}(S)u_t
=
\eta^\top
\mathbf{P}_{u_t}^\top
\mathcal R^{(r)}(S)
\mathbf{P}_{u_t}
\eta,
\label{eq:quadratic_input_term}
\end{align}
where we have used $u_t=\mathbf{P}_{u_t}\eta$ and
\begin{align}
&\mathbb{E}_w
\left[
B_p(w,\theta^{(r)})^\top S B_p(w,\theta^{(r)})
\right]
\nonumber\\
&=
\sum_{|\alpha|\leq p}
B_\alpha(\theta^{(r)})^\top
S
B_\alpha(\theta^{(r)})
=
\mathcal R^{(r)}(S).
\label{eq:quadratic_input_expectation}
\end{align}

For the second term in \eqref{eq:quadratic_moment_expansion}, the conditional independence of $w_t$ and $z_t$ given $\theta$ implies
\begin{align}
&\mathbb{E}
\left[
z_t^\top
A_p(w_t,\theta)^\top
S
B_p(w_t,\theta)u_t
\,\middle|\,
\theta=\theta^{(r)}
\right]
\nonumber\\
&=
\mathbb{E}
\left[
z_t
\,\middle|\,
\theta=\theta^{(r)}
\right]^\top
\mathbb{E}_w
\left[
A_p(w,\theta^{(r)})^\top
S
B_p(w,\theta^{(r)})
\right]
u_t
\nonumber\\
&=
\eta^\top
{\bGamma_t^{(r)}}^\top
\mathcal N^{(r)}(S)
\mathbf{P}_{u_t}
\eta,
\label{eq:quadratic_cross_term_1}
\end{align}
where Proposition~\ref{prop: first moment} gives
$
\mathbb{E}[z_t\mid\theta=\theta^{(r)}]
=
\bGamma_t^{(r)}\eta,
$
and the mode decomposition in \eqref{eq:PCE modes}, together with
\eqref{eq:w-basis-orthonormality}, gives
\begin{align}
&\mathbb{E}_w
\left[
A_p(w,\theta^{(r)})^\top
S
B_p(w,\theta^{(r)})
\right]
\nonumber\\
&=
\sum_{|\alpha|\leq p}
A_\alpha(\theta^{(r)})^\top
S
B_\alpha(\theta^{(r)})
=
\mathcal N^{(r)}(S).
\label{eq:quadratic_cross_expectation}
\end{align}
Similarly,
\begin{align}
&\mathbb{E}
\left[
u_t^\top
B_p(w_t,\theta)^\top
S
A_p(w_t,\theta)z_t
\,\middle|\,
\theta=\theta^{(r)}
\right]
\nonumber\\
&=
\eta^\top
\mathbf{P}_{u_t}^\top
\mathcal N^{(r)}(S)^\top
\bGamma_t^{(r)}
\eta.
\label{eq:quadratic_cross_term_2}
\end{align}

Combining the preceding expressions in
\eqref{eq:quadratic_moment_expansion} yields
\begin{align}
&\mathbb{E}
\left[
z_{t+1}^\top S z_{t+1}
\,\middle|\,
\theta=\theta^{(r)}
\right]
\nonumber\\
&=
\eta^\top
\Big[
\mathcal{M}_t^{(r)}
\!\left(
\mathcal Q^{(r)}(S)
\right)
+
{\bGamma_t^{(r)}}^\top
\mathcal N^{(r)}(S)
\mathbf{P}_{u_t}
\nonumber\\
&\qquad\quad
+
\mathbf{P}_{u_t}^\top
\mathcal N^{(r)}(S)^\top
\bGamma_t^{(r)}
+
\mathbf{P}_{u_t}^\top
\mathcal R^{(r)}(S)
\mathbf{P}_{u_t}
\Big]
\eta.
\label{eq:quadratic_moment_recursion_proof}
\end{align}
Comparing \eqref{eq:quadratic_moment_recursion_proof} with the definition of
$\mathcal M_{t+1}^{(r)}(S)$ in \eqref{eq: Mt def} gives
\eqref{eq:Mt_recursion}, which completes the proof.
\end{proof}
As a corollary of unrolling the recursion formula in Proposition \ref{prop: second moments}, we have the following identity:

\begin{corollary}
\label{cor:Mt_unrolled}
For a fixed quadrature node $\theta=\theta^{(r)}$ and any
$S\in\mathbb{S}^{n_z}$, define the iterated operators
\begin{equation}
\label{eq:Q_iterates}
\begin{aligned}
\mathcal{Q}_0^{(r)}(S)
&:=
S,
\\
\mathcal{Q}_{k+1}^{(r)}(S)
&:=
\mathcal{Q}^{(r)}
\left(
\mathcal{Q}_k^{(r)}(S)
\right),
\quad k\geq 0.
\end{aligned}
\end{equation}
Moreover, define
\begin{align}
\mathcal{D}_j^{(r)}(X)
:={}&
{\bGamma_j^{(r)}}^\top
\mathcal{N}^{(r)}(X)
\mathbf{P}_{u_j}
+
\mathbf{P}_{u_j}^\top
\mathcal{N}^{(r)}(X)^\top
\bGamma_j^{(r)}
\nonumber\\
&+
\mathbf{P}_{u_j}^\top
\mathcal{R}^{(r)}(X)
\mathbf{P}_{u_j}.
\label{eq:Djr_def}
\end{align}
Then, for $t=0,\ldots,H$, the conditional quadratic kernel admits
the unrolled representation
\begin{equation}
\label{eq:Mt_unrolled}
\mathcal{M}_t^{(r)}(S)
=
\mathbf{P}_z^\top
\mathcal{Q}_t^{(r)}(S)
\mathbf{P}_z
+
\sum_{j=0}^{t-1}
\mathcal{D}_j^{(r)}
\left(
\mathcal{Q}_{t-1-j}^{(r)}(S)
\right),
\end{equation}
where
$
\mathbf{P}_z
:=
\begin{bmatrix}
I_{n_z} & 0
\end{bmatrix},
$
and the summation in \eqref{eq:Mt_unrolled} is interpreted as zero
when $t=0$.
\end{corollary}

\begin{proof}
The case $t=0$ follows directly from \eqref{eq:M0_def}. Suppose that
\eqref{eq:Mt_unrolled} holds at time $t$. Applying the induction
hypothesis with $\mathcal{Q}^{(r)}(S)$ in place of $S$ and using
\eqref{eq:Mt_recursion} gives
\begin{align}
\mathcal{M}_{t+1}^{(r)}(S)
={}&
\mathbf{P}_z^\top
\mathcal{Q}_t^{(r)}
\left(
\mathcal{Q}^{(r)}(S)
\right)
\mathbf{P}_z
\nonumber\\
&+
\sum_{j=0}^{t-1}
\mathcal{D}_j^{(r)}
\left(
\mathcal{Q}_{t-1-j}^{(r)}
\left(
\mathcal{Q}^{(r)}(S)
\right)
\right)
+
\mathcal{D}_t^{(r)}(S)
\nonumber\\
={}&
\mathbf{P}_z^\top
\mathcal{Q}_{t+1}^{(r)}(S)
\mathbf{P}_z
+
\sum_{j=0}^{t}
\mathcal{D}_j^{(r)}
\left(
\mathcal{Q}_{t-j}^{(r)}(S)
\right),
\end{align}
which is \eqref{eq:Mt_unrolled} at time $t+1$.
\end{proof}

Corollary~\ref{cor:Mt_unrolled} allows the operator iterates
$\{\mathcal{Q}_k^{(r)}(S)\}_{k=0}^{H}$ to be computed first, and then
substituted into the finite sum in \eqref{eq:Mt_unrolled}.

\textbf{Recursive Condensation Algorithm:}
The preceding analyses provide recursive formulas for the
conditional first- and second-moment quantities at each quadrature node
$\theta^{(r)}$. Combining these recursions with the exact Gaussian
quadrature rule yields Algorithm~\ref{alg:recursive_moment_condensation} for constructing
the condensed matrices
$\{\mathbf{S},\bar{\mathbf{G}}_t,\mathbf{M}_t,\mathbf{c}_t\}$.

\begin{algorithm}[t]
\caption{Offline Condensed Matrices Computation}
\label{alg:recursive_moment_condensation}
\footnotesize
\begin{algorithmic}[1]
\Require PCE modes $\{A_\alpha(\theta),B_\alpha(\theta)\}_{|\alpha|\le p}$
\Require Quadrature nodes and weights $\{(\theta^{(r)},\omega_r)\}_{r=1}^{N_q}$
\Require Horizon $H$; output/cost matrices $C,Q,Q_f,R$; directions $\{a_t\}_{t=1}^{H}$
\Require Selector matrices $\mathbf P_U$ and $\{\mathbf P_{u_t}\}_{t=0}^{H-1}$
\Ensure $\mathbf{S},\{\bar{\mathbf{G}}_t,\mathbf{M}_t,\mathbf{c}_t\}_{t=1}^{H}$
\State Set $\mathbf{P}_z=[I_{n_z}\;0]$, $\bar R=I_H\otimes R$
\State Set $X_Q=C^\top QC$, $X_f=C^\top Q_fC$, and $X_{a_t}=C^\top a_ta_t^\top C$
\State Set $\mathcal X_s=\{X_Q,X_f,X_{a_1},\ldots,X_{a_H}\}$
\State Initialize $\bar{\mathbf{G}}_t=0$ and $\overline{\mathcal M}_t(X)=0$
\For{$r=1,\ldots,N_q$}
    \State Evaluate $A_\alpha(\theta^{(r)})$, $B_\alpha(\theta^{(r)})$
    \State Define the operator maps $\mathcal Q^{(r)}$, $\mathcal N^{(r)}$, and $\mathcal R^{(r)}$
    by \eqref{eq:Qtheta_def}--\eqref{eq:Rtheta_def}
    \State $\bGamma_0^{(r)}\gets\mathbf P_z$
    \For{$t=0,\ldots,H-1$}
        \State $\bGamma_{t+1}^{(r)}\gets A_{\mathbf 0}(\theta^{(r)})\bGamma_t^{(r)}+B_{\mathbf 0}(\theta^{(r)})\mathbf P_{u_t}$
    \EndFor
    \State $\bar{\mathbf G}_t\gets\bar{\mathbf G}_t+\omega_rC\bGamma_t^{(r)}$, $t=1,\ldots,H$
    \For{$X\in\mathcal X_s$}
        \State Form $\mathcal Q_k^{(r)}(X)$, $k=0,\ldots,H$, recursively:
        \Statex \hspace{\algorithmicindent}$\mathcal Q_0^{(r)}(X)=X,\quad
        \mathcal Q_{k+1}^{(r)}(X)=\mathcal Q^{(r)}(\mathcal Q_k^{(r)}(X))$
        \State Form $\mathcal M_t^{(r)}(X)$, $t=0,\ldots,H$, by
        \Statex \hspace{\algorithmicindent}$\begin{aligned}
        \mathcal M_t^{(r)}(X)={}&\mathbf P_z^\top\mathcal Q_t^{(r)}(X)\mathbf P_z\\
        &+\sum_{j=0}^{t-1}\mathcal D_j^{(r)}
        \big(\mathcal Q_{t-1-j}^{(r)}(X)\big)
        \end{aligned}$
        \State $\overline{\mathcal M}_t(X)\gets\overline{\mathcal M}_t(X)+\omega_r\mathcal M_t^{(r)}(X)$
    \EndFor
\EndFor
\State $\mathbf M_t\gets\overline{\mathcal M}_t(X_{a_t})$, $\mathbf c_t\gets\bar{\mathbf G}_t^\top a_t$
\State $\mathbf S\gets\mathbf P_U^\top\bar R\mathbf P_U+\sum_{t=1}^{H-1}\overline{\mathcal M}_t(X_Q)+\overline{\mathcal M}_H(X_f)$
\State \Return $\mathbf{S}$ and $\{\bar{\mathbf{G}}_t,\mathbf{M}_t,\mathbf{c}_t\}_{t=1}^{H}$
\end{algorithmic}
\end{algorithm}

Due to the polynomial structure of the PPKO matrices $A_p(w_t,\theta)$ and $B_p(w_t,\theta)$, one may use Gaussian quadrature to accurately compute the expectations in the condensed matrices. Specifically, the following theorem shows that the offline condensation algorithm computes the condensed matrices exactly with a suitable quadrature rule:
\begin{theorem}
If the quadrature rule in $\{(\theta^{(r)},\omega_r)\}_{r=1}^{N_q}$ is exact for all
polynomial integrands of degree at most $2Hp$, then Algorithm~\ref{alg:recursive_moment_condensation} exactly computes the condensed matrices $\mathbf{S}$ and $\{\bar{\mathbf{G}}_t,\mathbf{M}_t,\mathbf{c}_t\}_{t=1}^{H}$.
\end{theorem}
\begin{proof}
The nodewise computations follow directly from
Propositions~\ref{prop: first moment} and
\ref{prop: second moments}, together with
Corollary~\ref{cor:Mt_unrolled}. Moreover,
$\mathbb{E}[z_t\mid\theta]$ and
$\mathbb{E}[z_t^\top S z_t\mid\theta]$ are polynomials in each
component of $\theta$ of degree at most $tp$ and $2tp$,
respectively. Since $\eta$ does not depend on $\{w_t\}_{t=0}^{H-1}$ and $\theta$, this implies that the integrands in the condensed matrices have degree at most $2Hp$, where $H$ is the SMPC horizon. Thus, all required integrands are integrated exactly by the stated quadrature rule.
\end{proof}

\begin{remark}
If $\theta$ has independent components, then, since an $n$-point Gaussian quadrature rule is exact for all polynomial integrands with degree at most $2n-1$, $Hp+1$ nodes per component are sufficient. This results in a total of $(Hp+1)^{n_\theta}$ nodes in the tensor-product quadrature. Although this number grows exponentially with $n_\theta$, the condensed matrix computations are all done offline, which does not affect the online SMPC computation.

\end{remark}
Algorithm~\ref{alg:recursive_moment_condensation} therefore completes the exact offline construction of the condensed matrices for the finite-dimensional PPKO model. Importantly, no quadrature, Monte Carlo sampling, or stochastic Galerkin projection is performed over the full disturbance trajectory $\{w_t\}_{t=0}^{H-1}$. Instead, the PCE structure of the PPKO model and the IID structure of the disturbances are exploited to propagate the required moments recursively. Consequently, at each sampling instant, only $z_0=\Psi(x_0)$ must be updated before solving \eqref{eq:det_smpc_problem}, whose effective online decision dimension is $Hn_u$.

\section{Numerical Examples}
This section presents numerical results demonstrating the performance of the proposed PPKO-SMPC framework on two stochastic nonlinear dynamical systems: a Duffing oscillator subject to random forcing and a high-dimensional stochastic electrothermal cooling system governed by a coupled ODE--PDE model. These two distinct examples are selected because they exhibit nontrivial closed-loop regulation challenges. In particular, the electrothermal cooling system provides a high-dimensional case study that highlights the computational advantages of our condensed SMPC reformulation. In each case study, we compare the open-loop and closed-loop performance of four different MPC frameworks:
\begin{enumerate}
\item Deterministic EDMD-DL
\item Scenario-based SMPC 
\item Intrusive SGP-based SMPC 
\item PPKO-SMPC (our framework)
\end{enumerate}
The PPKO-SMPC framework and the numerical results can be reproduced using the code available at \url{https://github.com/wallytgy/Koopman-NN}. All computational examples were carried out on a Dell Latitude 7440 laptop equipped with a 13th Gen Intel Core i7-1365U processor and 32 GB of RAM, running Microsoft Windows 11 Enterprise. Reported solve times include online prediction/assembly and optimization, excluding offline setup and Monte Carlo audits.

\subsection{Stochastic Duffing Oscillator}
\label{subsec:duffing_setup}
Consider the controlled Duffing oscillator:
\begin{subequations}\label{eq:duffing_problem}
\begin{align}
\dot{x}_1 &= x_2, \label{eq:duffing_x1}\\
\dot{x}_2 &= -(\delta + 0.3w_{d,k})x_2 - \beta x_1 - \alpha x_1^3 + u_k + 1.5w_{f,k}. \label{eq:duffing_x2}
\end{align}
\end{subequations}
where $x = [x_1\;x_2]^\top \in \mathbb{R}^2$ denotes the states and $u_k \in \mathbb{R}$ is the control input. The parameters $\delta \sim \mathcal{U}[0,1]$, $\beta \sim \mathcal{U}[-2,2]$, and $\alpha \sim \mathcal{U}[0,2]$ are independent uncertain parameters following uniform distributions and remain fixed throughout each trajectory and prediction horizon. The disturbances $w_{d,k}$ and $w_{f,k}$ are i.i.d. $\mathcal N(0,1)$ random variables that are redrawn every $\Delta t = 0.02$ and held constant over each sampling interval.

\subsubsection{Model Identification for PPKO and Baselines}

\textbf{Data generation.}
The continuous-time dynamics in \eqref{eq:duffing_problem} were discretized with a fourth-order Runge--Kutta scheme using sampling interval $\Delta t=0.02$, over 200 time steps per trajectory. For model identification, we generated 400 stochastic trajectories, each with independently sampled uncertain parameters and disturbance realizations. For each trajectory, the initial condition was sampled from $\mathcal U[-1.5,1.5)\times\mathcal U[-1.5,1.5)$, and the inputs were sampled i.i.d. from $\mathcal N(0,0.6^2)$ and applied with zero-order hold over each control interval. The dataset was split at the trajectory level into 280/60/60 train/validation/test trajectories.

\textbf{Koopman observables.}
For both PPKO and the deterministic EDMD-DL baseline, Koopman observables were learned with EDMD-DL using a 28-dimensional lift: one constant observable, two coordinate observables, and 25 learned observables. The feature map was parameterized by a two-hidden-layer fully connected network with 64 neurons per hidden layer and $\tanh$ activations.

\textbf{PPKO identification.}
For PPKO, the polynomial operator maps $A_p(\theta,w)$ and $B_p(\theta,w)$ were represented with a joint total-degree-two Legendre--Hermite polynomial-chaos basis, using Legendre polynomials for the parametric uncertainty and Hermite polynomials for the disturbances. PPKO was trained by alternating minimization: the neural observable parameters were updated by Adam on minibatches, and the Koopman coefficient blocks $\{A_\alpha,B_\alpha\}$ were then updated once per epoch via closed-form ridge least squares on the full training set. Here, a batch size of 2048 refers to 2048 one-step transition samples $(x_k,u_k,\theta,w_k,x_{k+1})$ per Adam update, rather than 2048 trajectory segments. Training used early stopping based on validation prediction error, and the final selected ridge parameter was $10^{-6}$.

\textbf{Baseline models.}
For the deterministic EDMD-DL baseline, a separately learned dictionary of the same architecture was used, and Koopman operator matrices $\{A,B\}$ were fit by ridge regression. For the intrusive SGP-SMPC baseline, we used a total-degree-1 Legendre--Hermite expansion and solved the corresponding deterministic expanded PCE system. The intrusive SGP-SMPC and scenario SMPC baselines are equation-based and therefore do not require neural-network training. Scenario SMPC used 64 fixed parameter/disturbance paths during closed-loop regulation; its NLP and the SGP NLP were solved with IPOPT, while EDMD-DL used CLARABEL.

\subsubsection{Open-loop Performance}
To assess the accuracy of the proposed PPKO model in open-loop prediction, the uncontrolled Duffing oscillator was propagated over a horizon of $H=10$ steps with $\Delta t = 0.02$ from a representative initial condition. The PPKO mean and second moment were compared to Monte~Carlo (MC) ground truth ($N_{\mathrm{mc}} = 5{,}000$ samples) and the baseline models. Fig.~\ref{fig:open_loop_moments} shows that the PPKO predictions closely match the MC reference in both the mean and second-moment trajectories, which is crucial for our SMPC formulation since the constraints use second-moment predictions. Moreover, Fig.~\ref{fig:open_loop_distributions} shows that the PPKO state distributions for $H=10$ also closely match the MC histograms, confirming that the learned Koopman model accurately captures the stochastic open-loop dynamics.

In this case study, the deterministic EDMD-DL baseline fails to capture both the first and second moments accurately, demonstrating the importance of incorporating uncertainty for accurate open-loop prediction. The intrusive SGP baseline also accurately captures the state second moments and empirical distributions. However, the expanded system of PCE coefficients produces a large nonlinear system: the degree-one expansion contains 24 coefficients per state, resulting in a 48-dimensional system for the open-loop horizon considered here. This expansion makes the resulting closed-loop SMPC problem computationally expensive, as discussed in the next section.

\begin{figure}
    \centering
    
    \vspace{0.2cm}
    \includegraphics[width=\columnwidth]{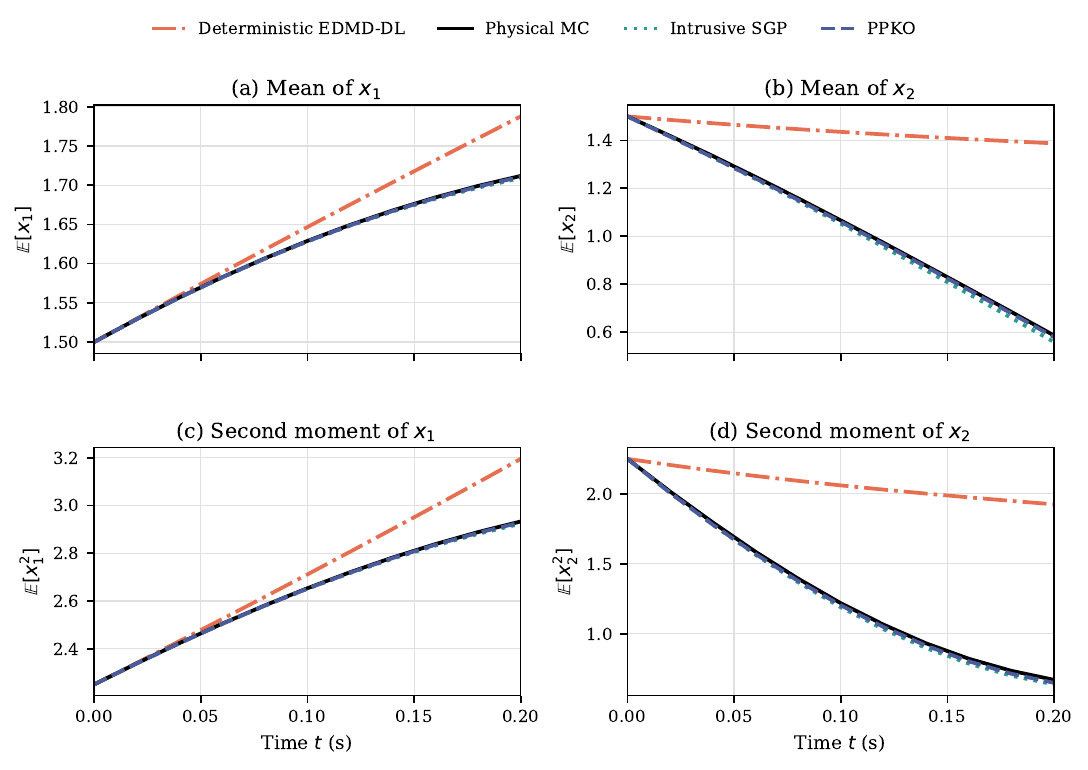}

    \vspace{-0.3cm}
    
    \caption{Open-loop Duffing oscillator state mean and second moments for the representative initial condition $x_0=[1.5,1.5]^\top$: mean and raw second-moment trajectories for the true system, PPKO, intrusive SGP, and deterministic EDMD-DL.}
    \label{fig:open_loop_moments}
\end{figure}

\begin{figure}
    \centering
    
    \vspace{0.2cm}
    \includegraphics[width=\columnwidth]{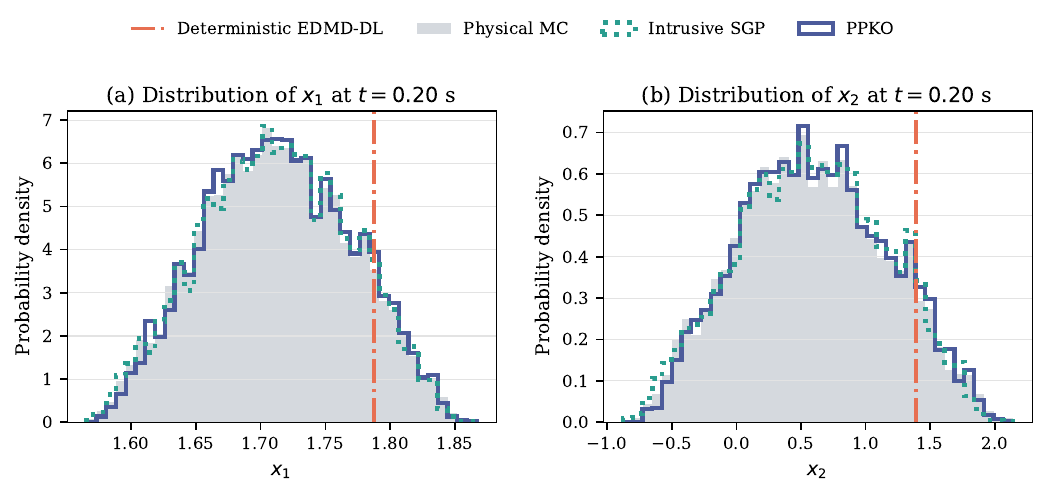}

    \vspace{-0.3cm}
    
    \caption{State distributions of the open-loop Duffing oscillator for $H=10$. The physical Monte~Carlo distributions are compared with distributions sampled from PPKO and intrusive SGP using the same $5{,}000$ parameter and disturbance realizations. The deterministic EDMD-DL predictions are shown as vertical lines.}
    \label{fig:open_loop_distributions}
\end{figure}

\subsubsection{Closed-loop Performance}

\label{subsec:duff_control}

The PPKO-SMPC controller \eqref{eq:det_smpc_problem} was then implemented in receding-horizon form with a prediction horizon of $H=5$. A tensor-product Gauss--Legendre quadrature rule with 11 nodes per uncertain parameter was used to evaluate the condensed matrices \eqref{eq:pcecost}--\eqref{eq:2ndmoment}. The SMPC LQR loss matrices were chosen as $Q=\operatorname{diag}(5,2)$, $Q_f=\operatorname{diag}(200,120)$, and $R=0.05I$. The control inputs were constrained by $-3\leq u_j\leq3$, the predicted state means by $-2\leq\mathbb{E}[x_{i,j}]\leq2$ for $i=1,2$, and the terminal raw second moments by $\mathbb{E}[x_{1,H}^2]\leq3.5$ and $\mathbb{E}[x_{2,H}^2]\leq3.5$. The SMPC dynamic optimization \eqref{eq:det_smpc_problem} was solved at each step using CVXPY \cite{diamond2016cvxpy,agrawal2018rewriting} with the CLARABEL solver \cite{Clarabel_2024} and warm-started by shifting the previous solution forward by one step. The true system was propagated with a fourth-order Runge--Kutta scheme at a sampling time of $\Delta t=0.02$~s.

To evaluate the efficacy of our proposed framework under parametric uncertainty and process disturbances, three representative parameter configurations were selected, each producing qualitatively distinct dynamical behavior, and ten disturbance realizations were applied to each configuration. In every case, the proposed stochastic Koopman--MPC scheme successfully steers the state to a neighborhood of the origin despite the persistent disturbances, demonstrating that the learned Koopman model captures sufficient dynamic structure to enable closed-loop stabilization across all three regimes. For qualitative comparison, uncontrolled open-loop trajectories under the same disturbance realizations and numerically computed equilibria along the line $x_2=0$ are overlaid in the phase portraits of Fig.~\ref{fig:control_duff}.
\begin{figure*}
    \centering
    
    \vspace{0.2cm}
    
    \includegraphics[width = 0.75\linewidth]{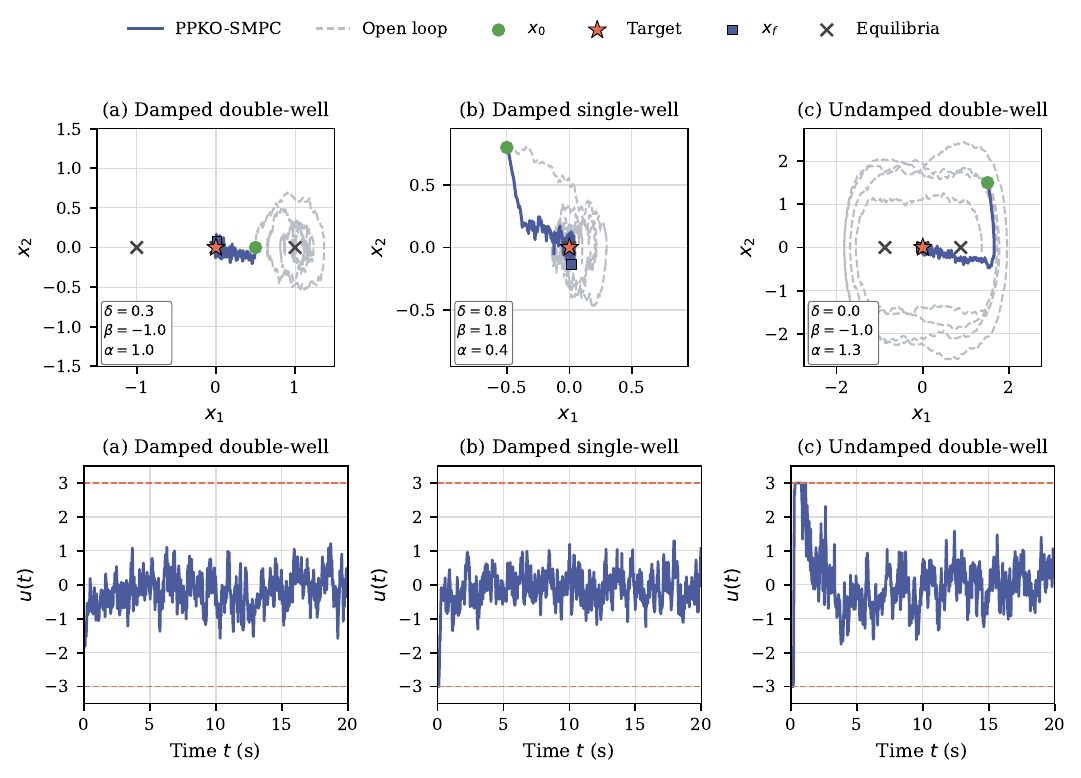}

    \vspace{-0.3cm}
    
    \caption{Closed-loop dynamics and control input profiles for 3 representative parametric uncertainty realizations corresponding to three different open-loop dynamical regimes: damped double-well (left), damped single-well (middle), and undamped double-well (right). Uncontrolled open-loop trajectories (dotted) are shown for comparison. $x_0$ is the initial point, $x_f$ is the final point of the simulation, and the target is the origin (0,0). Equilibria correspond to either a stable point or a saddle point.}
    \label{fig:control_duff}
\end{figure*}

\begin{table}[t]
\centering
\caption{Closed-loop performance metrics for the Duffing oscillator.}
\label{tab:duffing_closed_loop}
\scriptsize
\setlength{\tabcolsep}{2.5pt}
\resizebox{\columnwidth}{!}{%
\begin{tabular}{@{}lcccc@{}}
\hline
Controller & State RMS & Input RMS & Viol. & Mean solve [ms] \\
\hline
EDMD-DL       & 0.0576 & 0.635 & 0/60k & 1.64  \\
Scenario SMPC & 0.0580 & 0.650 & 0/60k & 4.73  \\
SGP-SMPC      & 0.0569 & 0.655 & 0/60k & 23.77 \\
PPKO-SMPC     & 0.0568 & 0.656 & 0/60k & 2.05  \\
\hline
\end{tabular}
}
\end{table}

To evaluate the closed-loop performance of our PPKO framework against the listed baselines, Table~\ref{tab:duffing_closed_loop} reports the state root-mean-square (RMS), input RMS, empirical constraint violations, and mean online solution time across 30 simulations using identical initial conditions, parameter configurations, and disturbance realizations for all controllers. The state RMS is computed over the final \(N_t=200\) samples of each of the \(N=30\) closed-loop trajectories as
\[
\mathrm{RMS}_x=\frac{1}{N}\sum_{r=1}^{N}\sqrt{\frac{1}{N_t}\sum_{k=N_{\mathrm{sim}}-N_t+1}^{N_{\mathrm{sim}}}\left\|x_k^{(r)}\right\|_2^2},
\]
where \(N_{\mathrm{sim}}=1000\). The input RMS is the mean of the per-trajectory RMS values computed using all \(N_{\mathrm{sim}}\) control updates. None of the controllers violate any state constraints empirically in this case study. Notably, although all four controllers achieve successful regulation with comparable state RMS values and similar control effort, our PPKO-SMPC framework achieves a mean solution time comparable to that of deterministic EDMD-DL. In contrast, scenario-based SMPC requires more than twice the mean solution time, while intrusive SGP-SMPC, which solves a nonlinear program involving the expanded system of PCE coefficients, is substantially slower. These results empirically demonstrate the significant computational advantage of learning PCE parameterizations of the Koopman matrices and leveraging the condensing methods in Section~\ref{sec: condensed} to mitigate the online curse of dimensionality when solving the SMPC QP.

\subsection{Stochastic Electrothermal Cooling System}
\label{subsec:electrothermal_setup}
Consider the following stochastic electrothermal PDE--ODE benchmark,
based on heat diffusion and Joule heating adapted from~\cite{casper2019automated,ali2010electrothermal}:
\begin{subequations}\label{eq:electrothermal_problem}
\begin{align}
\partial_t T
&= 0.02(1+0.25\theta_a)\partial_{ss}T-0.1T \nonumber\\
&\quad +0.2\sum_{j=1}^{8}\phi_j(s)R_j[T,\theta_R]i_j^2
-\sum_{\ell=1}^{10}b_\ell(s)u_{\ell,k},\\
\dot{i}
&= -\operatorname{diag}(d)(i-\bar{i})+Gw_k,\\
T(0,t)&=T(1,t)=0.
\end{align}
\end{subequations}
Here, $T(s,t)$ is the temperature profile with spatial domain $s\in[0,1]$,
$i\in\mathbb{R}^{8}$ is the vector representing the heater currents, $u_k := [u_{1,k}, \cdots u_{10,k}]^\top \in\mathbb{R}^{10}$ is the control action containing the cooling inputs, and $w_k \in \mathbb{R}^2$ are the current disturbances. The control actions $u_k$ and disturbances $w_k$ are applied with zero-order hold and remain constant over every sampling interval $k$. The heater and cooling profiles are Gaussian-shaped with parameters
\begin{align}
\phi_j(s)&=\exp\left(-\frac{(s-j/9)^2}{2(0.05)^2}\right),\\
b_\ell(s)&=\exp\left(-\frac{(s-\ell/11)^2}{2(0.06)^2}\right).
\end{align}
The current parameters are $d_j=0.5+(j-1)/7$,
$\bar{i}_j=0.8+0.1\cos(2\pi j/9)$, $G_{j1}=2$, and
$G_{j2}=2\cos(2\pi j/9)$. The heater resistance has a locally averaged temperature dependence given by
\begin{align}
R_j[T,\theta_R]&=(1+0.3\theta_R)\exp(0.4\bar{T}_j),\\
\bar{T}_j&=\frac{\int_0^1\phi_j(s)T(s,t)\,\mathrm{d}s}
{\int_0^1\phi_j(s)\,\mathrm{d}s}.
\end{align}
The uncertainty parameters $\theta_a,\theta_R\sim\mathcal{U}[-1,1]$
are jointly independent and remain fixed throughout each trajectory.
The components of the disturbances $w_k$ are i.i.d.\
$\mathcal{U}[-1,1]$ variables that are independent of the uncertainty
parameters. They are redrawn every $\Delta t=0.1$ and held constant within
each control interval. Discretizing the PDE at 64 interior nodes yields
72 physical states.

\subsubsection{Model Identification for PPKO and Baselines}
To avoid repetition, we omit definitions and some implementation
details already introduced in the previous case study.

\textbf{Data generation.}
Both learned models were initialized using 5,760/1,440/1,440 independent five-step
train/validation/test trajectories with $\Delta t=0.1$, including post-surge current levels. Both were then refined using the same 32,768 independent
one-step transitions and 4,000 five-step training trajectories, with common
1,500-trajectory validation and test datasets.
Simulation used Strang splitting between the linear
diffusion/ambient-loss and nonlinear heating/cooling dynamics,
with RK4 for the latter and four substeps per sampling interval.

\textbf{Koopman observables.}
Both PPKO and deterministic EDMD-DL used 97-dimensional lifts
comprising one constant, 72 state coordinates, and 24 learned
observables, with $\tanh$ activation layers.

\textbf{PPKO identification.}
As in the Duffing study, we parametrized $A_p(\theta,w)$ and
$B_p(\theta,w)$ using a joint total-degree-two Legendre basis in
$(\theta,w) = (\theta_a,\theta_R,w_1,w_2)$ and trained the PPKO models using
the alternating minimization method described in
\cite{li2017extended}, with a ridge regularization of $10^{-6}$ for the PCE
coefficient matrices. Both PPKO and EDMD-DL were further fine-tuned using the same
five-step mean-squared prediction loss over all 72 physical states to improve prediction accuracy over the SMPC horizon.

\textbf{Baseline models.}
EDMD-DL used the same datasets and five-step refinement loss as PPKO, with constant Koopman matrices $A$ and $B$ and a ridge regularization parameter of $10^{-4}$ for the initial operator fit.
Intrusive SGP used a degree-one Legendre expansion with a total of 13 coefficients per physical state.
Both intrusive SGP and scenario-based SMPC are equation-based
and required no neural-network training. Scenario SMPC used 128 fixed parameter/disturbance paths.

\begin{figure}
    \centering
    
    \vspace{0.2cm}
    \includegraphics[width=\columnwidth]{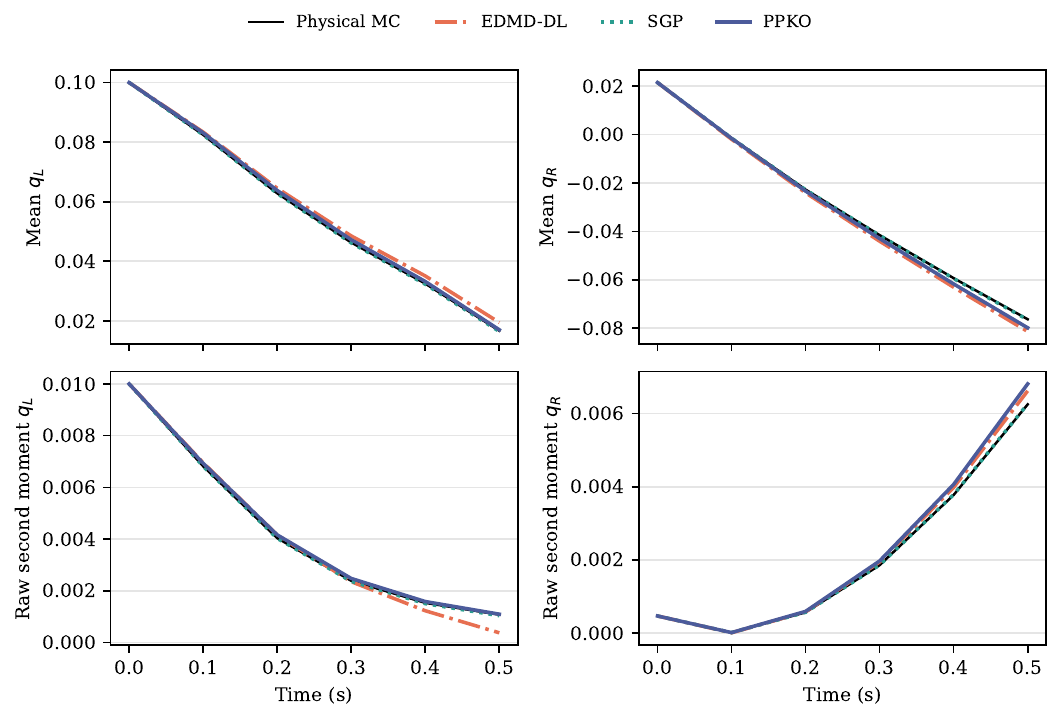}

    \vspace{-0.3cm}
    
    \caption{Open-loop mean and raw second-moment trajectories of the regional temperatures $q_L$ and $q_R$ for the electrothermal system under prescribed cooling inputs.}
    \label{fig:electrothermal_open_loop_moments}
\end{figure}

\begin{figure}
    \centering
    
    \vspace{0.2cm}
    \includegraphics[width=\columnwidth]{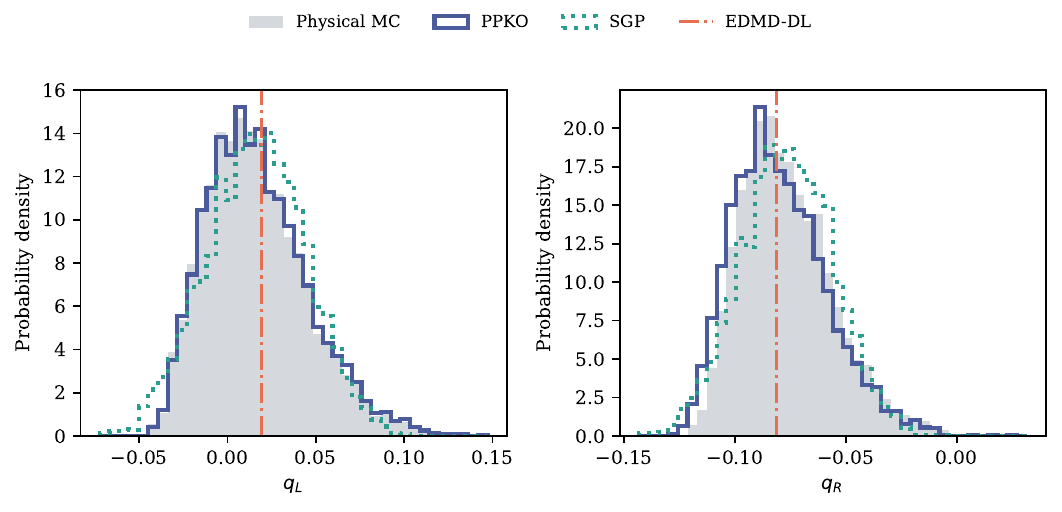}

    \vspace{-0.3cm}
    
    \caption{Regional-temperature distributions at $H=5$ ($t=0.5$), using $5{,}000$ matched parameter and disturbance realizations. Deterministic EDMD-DL predictions are shown as vertical lines.}
    \label{fig:electrothermal_open_loop_distributions}
\end{figure}

\subsubsection{Open-loop Performance}
Similar to the previous case study, to assess the accuracy of the
proposed PPKO model in open-loop prediction, the electrothermal
cooling system \eqref{eq:electrothermal_problem} was propagated
under prescribed cooling inputs over a horizon of $H=5$ steps
with $\Delta t=0.1$ from a representative initial condition.
We consider two linear outputs, $q_L$ and $q_R$, representing
the average temperatures over the left and right halves of
the spatial domain, respectively:
\begin{equation}
q_L(t)=2\int_0^{1/2}T(s,t)\,\mathrm{d}s,
\qquad
q_R(t)=2\int_{1/2}^{1}T(s,t)\,\mathrm{d}s.
\end{equation}
The PPKO mean and second moments of $q_L$ and $q_R$ were compared
to Monte~Carlo (MC) ground truth ($N_{\mathrm{mc}}=32{,}768$ samples)
and the baseline models in Fig.~\ref{fig:electrothermal_open_loop_moments}. The PPKO moments track the ground truth
moments well over the prediction horizon, but there is a small
deviation in $\mathbb{E}[q_R^2]$. This is consistent with
Remark~\ref{remark: structural bias}, which notes that structural
bias can arise in PPKO and universal approximation is not guaranteed.
Nonetheless, Fig.~\ref{fig:electrothermal_open_loop_distributions}
shows that the PPKO output distributions at $H=5$ approximately
match the MC histograms, which empirically demonstrates that PPKO approximates the stochastic dynamics well.

Fig.~\ref{fig:electrothermal_open_loop_moments} also shows that
EDMD-DL does not track the moment trajectories as well as the stochastic
baselines. While intrusive SGP accurately captures the moment dynamics,
its degree-one expansion shows minor discrepancies in the distributions
in Fig.~\ref{fig:electrothermal_open_loop_distributions}.
The expanded nonlinear PCE coefficient system has a total dimension
of 936, making open-loop inference and closed-loop optimization
more computationally expensive.
\subsubsection{Closed-loop Performance}
To avoid repetition, we omit definitions and implementation
details already introduced in the previous case study. The control objective in this case study is to regulate the temperature profile to zero, and the PPKO-SMPC controller
\eqref{eq:det_smpc_problem} using the PPKO model was implemented
in receding-horizon form with $H=5$ and $\Delta t=0.1$.
The state weighting matrices were
$Q=\operatorname{diag}(hI_{64},0_{8\times8})$ and $Q_f=6Q$,
where $h=1/65$ is the spatial discretization parameter, and the input penalty matrix was $R=0.05I_{10}$.
The cooling-input constraints are $0\leq u_{\ell,k}\leq2$.
For the stochastic controllers, both output raw second moments were constrained at every time step using the second-moment constraints in the SMPC problem \eqref{eq:ppko-2ndorderconstraints}. The second-moment constraints are $\mathbb{E}[q_L^2],\mathbb{E}[q_R^2] \leq 0.0121$ for time indices below 20, $\mathbb{E}[q_L^2],\mathbb{E}[q_R^2] \leq 0.006$ thereafter, and $\mathbb{E}[q_L^2],\mathbb{E}[q_R^2] \leq 0.011$ over indices 31--42 and 81--92 following the surges (see Fig.~\ref{fig:electrothermal_conditional_moments} for an illustration). To account for second-moment prediction errors, PPKO tightens the second-moment bounds by $0.0002$ for indices $\geq 20$. The condensed PPKO QCQP and EDMD-DL QP were both solved with CVXPY and CLARABEL.

% The audit used the common untightened bounds and fixed design prior, without injecting future surges; violation counts covered $180\times5\times2=1800$ update--lead--output checks.

To evaluate our PPKO framework against the other baselines, we performed the regulation task with the moment constraints on $q_L$ and $q_R$,
subject to sudden current surge events that added $0.8$ to all eight currents immediately before optimization at $k=30,80$ ($t=3,8$).
Fig.~\ref{fig:electrothermal_closed_thermal_rms} shows the thermal RMS
over 64 trajectories, where PPKO and the stochastic baselines achieve
smaller thermal RMS than deterministic EDMD-DL, especially during the periods following the current surges.

To evaluate each model's ability to enforce second-moment constraints, we propagate, at every control update of a representative trajectory, the controller's
input sequence through the physical model \eqref{eq:electrothermal_problem} using 2,048 Monte Carlo samples.
Fig.~\ref{fig:electrothermal_conditional_moments} shows the
Monte Carlo estimates of $\mathbb{E}[q_L^2]$ and
$\mathbb{E}[q_R^2]$ at the end of each controller's prediction horizon.
We observe that the PPKO framework and stochastic baselines achieve
comparable regulation with better adherence to the second-moment constraints
than deterministic EDMD-DL, illustrating the benefit of incorporating
stochastic dynamics in constraint enforcement.
Moreover, PPKO does not violate any physical second-moment constraints in this representative trajectory, while the stochastic baselines exhibit violations during the current surges. The EDMD-DL baseline also strongly violates the physical second-moment constraints during the current surges. Unlike the previous Duffing case study, this case study illustrates the importance of incorporating uncertainty propagation for moment-based constraints in SMPC regulation for stochastic systems.

Similar to the previous case study,
Table~\ref{tab:electrothermal_comparison} reports representative-trajectory metrics from a single closed-loop simulation.
% , whereas the thermal RMS figure aggregates the thermal RMS from 64 closed-loop simulation trajectories 
The table shows the thermal RMS, input RMS, solve times, and physical second-moment constraint
violations for PPKO and the baselines. We observe that PPKO achieves regulation
performance comparable to the stochastic baselines, with shorter solve
times than the intrusive SGP and scenario controllers, while exhibiting
no physical second-moment constraint violations. Under the common
CLARABEL solver, PPKO is slower than deterministic EDMD-DL because its
condensed problem includes second-moment constraints. The lower PPKO
solve time relative to the equation-based stochastic baselines illustrates
the computational benefit of offline condensation in this higher-dimensional case study.

\begin{table}[t]
\centering
\caption{Closed-loop performance metrics for a representative
trajectory of the electrothermal system.}
\label{tab:electrothermal_comparison}
\scriptsize
\setlength{\tabcolsep}{2.5pt}
\resizebox{\columnwidth}{!}{%
\begin{tabular}{@{}lcccc@{}}
\hline
Controller & Thermal RMS & Input RMS & Viol. & Mean solve [ms] \\
\hline
EDMD-DL & 0.0306 & 0.1364 & 97/1800 & 4.72 \\
Scenario SMPC & 0.0155 & 0.1367 & 23/1800  & 289.40 \\
SGP-SMPC      & 0.0155 & 0.1367 & 24/1800  & 72.31  \\
PPKO-SMPC & 0.0162 & 0.1372 & 0/1800 & 18.26 \\
\hline
\end{tabular}
}
\end{table}

\begin{figure}
    \centering
    
    \vspace{0.2cm}
    \includegraphics[width=\columnwidth]{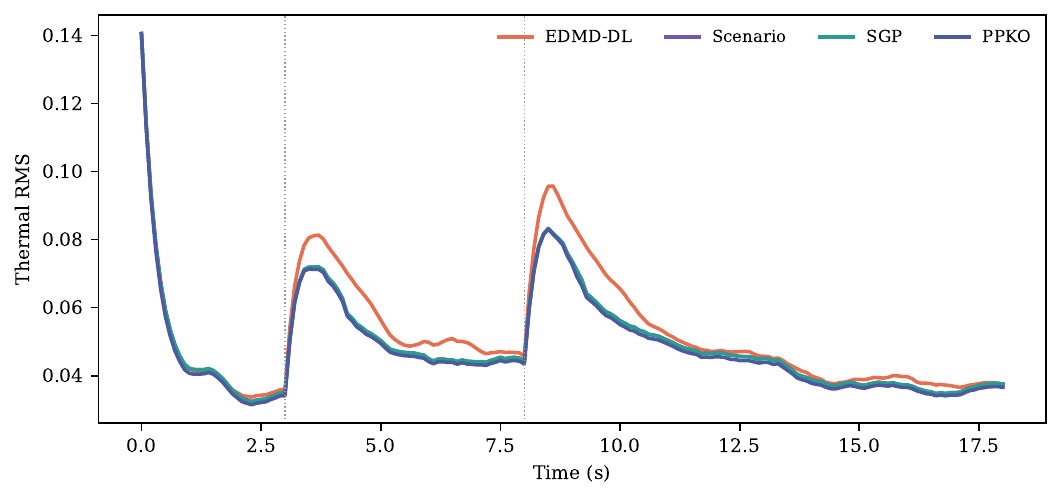}

    \vspace{-0.3cm}
    
    \caption{Closed-loop thermal RMS over 64 matched trajectories for PPKO and the baselines under current surges.}
    \label{fig:electrothermal_closed_thermal_rms}
\end{figure}

\begin{figure}
    \centering
    
    \vspace{0.2cm}
    \includegraphics[width=\columnwidth]{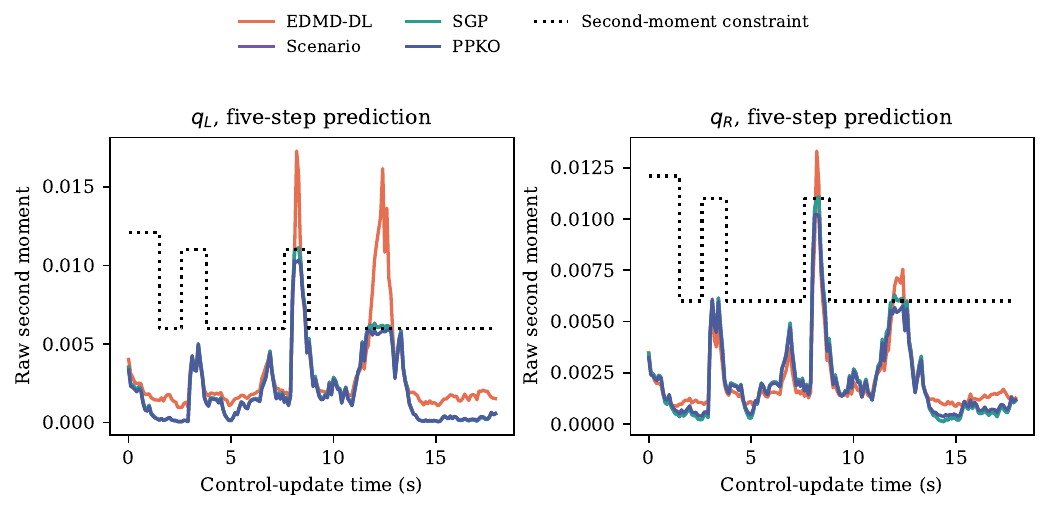}

    \vspace{-0.3cm}
    
    \caption{Physical Monte Carlo second moments of $q_L$ and $q_R$ at the prediction horizon under each optimized control sequence. The dotted curve denotes the second-moment constraint.}
    \label{fig:electrothermal_conditional_moments}
\end{figure}

\subsection{Discussion}
The two distinct numerical case studies demonstrate that our proposed PPKO-SMPC framework provides a favorable tradeoff between stochastic modeling accuracy and online computational efficiency. In both case studies, PPKO captures the first- and second-order statistics of the nonlinear stochastic dynamics with substantially better accuracy than the deterministic EDMD-DL model and with comparable accuracy to the intrusive SGP and Monte Carlo baselines. The closed-loop regulation simulations show that PPKO-SMPC achieves performance comparable to the stochastic baselines, with online computation substantially below the equation-based stochastic baselines, although above deterministic EDMD-DL for the electrothermal QCQP under the common solver. This computational advantage becomes particularly pronounced for the high-dimensional electrothermal PDE–ODE example, which illustrates the benefit of the condensed SMPC formulation as the number of state dimensions, lifted dimensions, and PCE basis functions increases. Furthermore, deterministic EDMD-DL may violate stochastic constraints during closed-loop operation, which motivates the use of SMPC formulations in stochastic settings. Overall, these results indicate that the proposed framework can incorporate PCE-based uncertainty propagation and Koopman operators into nonlinear SMPC without inheriting the curse of dimensionality typically associated with conventional stochastic formulations.

\section{Conclusion}
In summary, we developed the PPKO framework, integrating Koopman operator theory and PCEs for SMPC of nonlinear systems with parametric uncertainty and stochastic disturbances. We showed that PCE-parameterized Koopman matrices introduce no irreducible approximation gap, with the additional error vanishing as the PCE degree increases. We also developed a recursive offline condensation scheme that computes the required moment matrices and yields a deterministic SMPC problem whose online dimension depends only on the horizon and control dimension. The framework is validated on numerical examples spanning different dynamics and problem sizes.

Several important research directions remain beyond the scope of this article. These include analyses of recursive feasibility \cite{hewing2020recursively}, closed-loop stability \cite{cannon2009model}, adapting the parametric uncertainty distribution online \cite{zhang2020adaptive}, and handling temporally correlated disturbances, which are all directions that require further investigation. We deliberately focused here on developing the PPKO framework, including its approximation capabilities and condensed deterministic SMPC formulations, leaving these next theoretical developments as future work.

\begin{ack}
This research was supported by the U.S. Food and Drug Administration under the FDA BAA-22-00123 program, Award Number 75F40122C00200, and by the U.S. DOE, Office of Science, ASCR program under the Scientific Discovery through Advanced Computing (SciDAC) Institute ``LEADS: LEarning-Accelerated Domain Science.'' Efstathios Iliakis was further supported by the Onassis Foundation Scholarships' Program. Wallace Gian Yion Tan was further supported by the MathWorks Chemical Engineering Fellowship.
\end{ack}

\bibliographystyle{plain}
\bibliography{ref}

\end{document}